\documentclass[12pt]{article}

\usepackage{bbm}
\usepackage[bbgreekl]{mathbbol}
\usepackage[T1]{fontenc}
\usepackage{manfnt,amsmath,mathrsfs,amsfonts,amssymb,amsthm,graphicx,tikz,stmaryrd,cite}
\usepackage{tikz-cd}
\usetikzlibrary{matrix,arrows,decorations.pathmorphing}
\usepackage[hidelinks]{hyperref}

\hypersetup{backref=true,       
	pagebackref=true,               
	hyperindex=true,                
	colorlinks=true,                
	breaklinks=true,                
	urlcolor= black,                
	linkcolor= black,                
	bookmarks=true,                 
	bookmarksopen=false,
	filecolor=black,
	citecolor=blue,
	linkbordercolor=black
}

\AtBeginDocument{\hypersetup{pdfborder={0 0 1}}}

\usepackage[nottoc,numbib]{tocbibind}
\usepackage{esint}

\newtheorem{thm}{Theorem}[section]
\newtheorem{lem}[thm]{Lemma}
\newtheorem{prop}[thm]{Proposition}
\newtheorem{cor}[thm]{Corollary}

\theoremstyle{definition}
\newtheorem{defn}[thm]{Definition}

\renewcommand{\bar}{\overline}
\renewcommand{\tilde}{\widetilde}
\renewcommand{\hat}{\widehat}

\DeclareMathSymbol\bbDelta \mathord{bbold}{"01}

\DeclareMathOperator{\WF}{WF}

 \newcommand{\cD}{\mathcal{D}}
\newcommand{\cE}{\mathcal{E}} \newcommand{\cF}{\mathcal{F}}
 \newcommand{\cH}{\mathcal{H}}
 
\newcommand{\cK}{\mathcal{K}} \newcommand{\cL}{\mathcal{L}}
 \newcommand{\cN}{\mathcal{N}}
\newcommand{\cO}{\mathcal{O}} \newcommand{\cP}{\mathcal{P}}
 
\newcommand{\cS}{\mathcal{S}} \newcommand{\cT}{\mathcal{T}}
 \newcommand{\cV}{\mathcal{V}}

\newcommand{\fC}{\mathfrak{C}} \newcommand{\fD}{\mathfrak{D}}

 \newcommand{\fX}{\mathfrak{X}}
\newcommand{\fY}{\mathfrak{Y}} 

\newcommand{\fg}{\mathfrak{g}}

\newcommand{\fm}{\mathfrak{m}}

\newcommand{\fu}{\mathfrak{u}}

\newcommand{\bC}{\mathbb{C}}

 \newcommand{\bR}{\mathbb{R}}

 \newcommand{\bZ}{\mathbb{Z}}

\DeclareMathOperator{\Div}{div}

\DeclareMathOperator{\SO}{SO}
\DeclareMathOperator{\Or}{O}
\DeclareMathOperator{\SU}{SU}
\DeclareMathOperator{\Un}{U}

\newcommand{\fso}{\mathfrak{so}}

\DeclareMathOperator{\Ad}{Ad}

\DeclareMathOperator{\id}{id}

\DeclareMathOperator{\Tr}{Tr}

\DeclareMathOperator{\diag}{diag}

\DeclareMathOperator{\Span}{Span}

\DeclareMathOperator{\supp}{supp}

\DeclareMathOperator{\Gr}{Gr}

\DeclareMathOperator{\Vol}{Vol}
\DeclareMathOperator{\FE}{FE}
\DeclareMathOperator{\CD}{CD}
\DeclareMathOperator{\Hol}{Hol}

\DeclareMathOperator{\singsupp}{singsupp}

\begin{document}

\title{A Trace Formula on Stationary Kaluza-Klein Spacetimes}
\date{}
\author{Anthony McCormick}
\maketitle

\begin{abstract}
	We prove relativistic versions of the ladder asymptotics from \cite{guillemin1990reduction} on principal bundles over globally hyperbolic, stationary, spatially compact spacetimes equipped with a Kaluza-Klein metric. This involves understanding the distribution of the frequency spectrum for the wave equation on a Kaluza-Klein spacetime when restricted to the isotypic subspace of an irreducible representation of the structure group, in the limit that the weight of the representation approaches infinity in the Weyl chamber. This is a direct generalization of the results from \cite{strohmaier2020semi} and is closely related to \cite{strohmaier2018gutzwiller}, \cite{islam2021gutzwiller}. Furthermore we show how to apply these results to frequency asymptotics for the massive Klein-Gordon equation on vector bundles as one takes the representation defining the vector bundle to infinity.
\end{abstract}

\tableofcontents

\section{Introduction}
Given a principal $G\subseteq \SO(k)$-bundle $P$ with connection $\omega$ over a Lorentzian manifold $(M^{n+1},g)$ one can form the Kaluza-Klein metric on $P$:
\[ g_{\omega} := \pi^*g + \Tr(\omega(-)\omega(-)^T) \]
where $\pi:P\to M$ is the projection map. In the special case where $M$ is globally hyperbolic, stationary and spatially compact we will see in Section \ref{section_2} that the Kaluza-Klein spacetime takes the form
\[ P=\bR_t\times P_0, \ M=\bR_t\times \Sigma_0, \ P_0\to \Sigma_0 \mbox{ a principal bundle} \]
with metric given by
\[ g_{\omega} = -((N\circ \pi)^2 - |\pi^*\eta|^2_{\pi^*h})dt^2 + dt\otimes (\pi^*\eta) + (\pi^*\eta)\otimes dt + \pi^*h + \Tr(\omega(-)\omega(-)^T) \]
where $h$ is a Riemannian metric on $\Sigma_0$, $N:\Sigma_0\to \bR_{>0}$ a smooth function, and $\eta$ a 1-form on $\Sigma_0$ satisfying $|\eta|_h^2 <N^2$ pointwise. Given an integral coadjoint orbit $\cO$ for $G$ we have irreducible unitary representations corresponding to $m\cO$ for each $m\in\bZ_{\geq 1}$ and as discussed in Section \ref{section_3} we get isotypic subspaces of the space of solutions to the wave equation with respect to the Kaluza-Klein metric:
\[ \cH_m \subseteq \ker\Box_{\omega} \ \mbox{ corresponding to the representation } m\cO. \]
As is shown in Section \ref{vector_bundles}, for $m$ sufficiently large the operator
\[ D_Z := \frac{1}{i}\partial_t \]
is self adjoint on $\cH_m$ with respect to the energy form with a discrete set of eigenvalues
\[ \cdots \leq \lambda_{m,\ell}\leq \lambda_{m,\ell+1}\leq \cdot \]
accumulating at $\pm\infty$ with at worst polynomial growth. For $E>0$ fixed, our goal is to study the following question. \\

{\bf Question:} \\
What are the $m\to\infty$ asymptotics of the frequency spectrum of $D_Z-mE$ on $\cH_m$? \\

This is a relativistic analogue of the question studied in \cite{guillemin1990reduction}, and is a generalization to non-trivial principal bundles with arbitrary compact structure groups of the results in \cite{strohmaier2020semi}. \\

We make precise the notion of the distribution of the frequency spectrum about $mE$ via the tempered distribution $\mu(E,m,-)$ on $\bR$ given by:
\[ \mu(E,m,\varphi) := \sum_{\ell\in\bZ} \varphi(\lambda_{m,\ell}-mE) \]
and we will study its asymptotics via the periodic generating function
\[ \Upsilon(\varphi)(\theta) := \sum_{m=1}^{\infty} \mu(E,m,\varphi)e^{im\theta} \in \cD'(S^1)=\cD'(\bR/2\pi\bZ). \]
The fact that $\mu(E,m,-)$ is tempered and $\Upsilon(\varphi)$ is defined on all of $C^{\infty}(S^1)$ are proven in Sections \ref{vector_bundles}, \ref{section_4} respectively. \\

To state our main results we need to make a dynamical assumption akin to the ``clean intersection hypotheses'' that appear in \cite{strohmaier2020semi}, \cite{guillemin1989circular} and \cite{guillemin1990reduction}. For this we start by introducing the notation:
\begin{itemize}
	\item $\cN$ is the symplectic manifold of affinely parametrized inextendible future-directed null geodesics on $P$ modulo the action of translation in the affine parameter,
	\item $\cN_{\cO}$ is the symplectic reduction of $\cN$ along the coadjoint orbit $\cO$ as in \cite{guillemin1990reduction},
	\item $\tilde{H}_Z:\cN_{\cO}\to \bR$ and $\tilde{\Phi}^Z_t$ are respectively the Hamiltonian and Hamiltonian flow corresponding to the reduction of the flow on $\cN$ induced by the flow of $Z=\partial_t$ on $P$.
\end{itemize}
The \textbf{clean intersection hypothesis} then states that $E>0$ is a regular value for $\tilde{H}_Z$ and the fibered product $\fY_E$ of the flow map
\[ \bR\times \tilde{H}_Z^{-1}(E) \to \tilde{H}_Z^{-1}(E) \]
with the diagonal map $\tilde{H}_Z^{-1}(E)\to \tilde{H}_Z^{-1}(E)\times \tilde{H}_Z^{-1}(E)$ is a clean fibered product (it is a smooth manifold with tangent spaces given by the non-necessarily-transverse fibered products of the respective tangent spaces of the factors). We now state our main theorems. These are completely analogous to the main theorems in \cite{guillemin1990reduction} and are direct relativistic generalizations of these.
\begin{thm}{\label{thm_1}}
	The wave front set of $\Upsilon(\varphi)
	\in \cD'(S^1)$ is contained in:
	\begin{align*} \fD_{\varphi,E}:= \Big\{ (\omega,r)\in S^1\times \bR_{>0} \
		& : \ \exists (T,\gamma))\in \fY_E  \mbox{ with } T\in\supp\hat{\varphi} \\ &\mbox{ such that }
		\omega=\Hol_{\cO}([0,T]\ni t\mapsto
		\tilde{\Phi}^Z_t(\gamma))\Big\} \end{align*}
	where this holonomy is taken with respect to a natural $\Un(1)$-bundle with connection over $\cN_{\cO}$ defined in \ref{connection}.
\end{thm}

\begin{thm}{\label{thm_2}}
	Let $n+1=\dim(M)$. Under the clean intersection hypothesis, $\fD_{\varphi,E}$
	is a union of the positive parts of finitely many fibers of $T^*S^1$, and
	$\Upsilon(\varphi) \in I^{n+\ell-1+\frac{1}{4}}(S^1;\fD_{\varphi,E})$
	where $2\ell:=\dim \cO$. Furthermore, we obtain an asymptotic expansion as $m\to \infty$:
	\[ \mu(E,m,\varphi) \sim \sum_{k=0}^{\infty} m^{n+\ell-1-k}
	a_k(\varphi,m)\]
	with each $a_k(\varphi,m)$ a distribution in $\varphi$,
	bounded in $m$ for $k,\varphi$ fixed, and
	\[ a_0(\varphi,m) = C_{n,d}\varphi(0)\Vol\left(\tilde{H}_Z^{-1}(E)\right)\]
	where by $\Vol$ we mean to take the invariant measure on the energy hypersurface.
\end{thm}

\begin{thm}{\label{thm_3}}
	Suppose that, in addition to the clean intersection
	hypothesis, we assumed that $0\notin\supp(\varphi)$ and there existed only finitely many
	non-degenerate periodic orbits $(T_1,\gamma_1),\cdots,(T_q,\gamma_q)\in \fY_E$ with each $T_j\neq 0$. Then we actually
	obtain a better asymptotic expansion as $m\to \infty$:
	\[ \mu(E,m,\varphi)\sim \sum_{k=0}^{\infty} m^{-k}
	a_k(\varphi,m)\]
	and $a_0(\varphi,m)$ is of the form:
	\[ a_0(\varphi,m) = C_{n,d}\sum_{j=0}^q \Hol_{\cO}(T_j,\gamma_j)^m \frac{T_j^{\#}}{2\pi}
	\hat{\varphi}(T_j)\frac{ e^{i\pi
			\fm_j/4}}{|\det(I-P_j)|^{1/2}} .\]
	Where $T_j^{\#}$ is the minimum positive value of $T$ such that $\tilde{\Phi}^Z_T(\gamma_j)=\gamma_j$, $P_j$ is the linearized Poincar\'e first return map of $\gamma_j$, and $\fm_j$ is the Conley-Zehnder index.
\end{thm}

These theorems are proven in Section \ref{section_4} using the tools developed in the previous section. From the results of Section \ref{vector_bundles} we also obtain a corollary concerning the frequency distribution of $D_Z$ for vector bundles. Let's describe this now. \\

Let $\cV_m\to M$ denote the Hermitian vector bundle associated to the representation $\kappa_m$ equipped with the covariant derivative induced from $\omega$ on $P$ and the representation. We then have a massive Klein-Gordon operator $\Box_m$ with mass given by the eigenvalue
\[ \langle m\Lambda_0, \ m\Lambda_0 + \rho\rangle \]
of the quadratic Casimir on the representation corresponding to $m\cO$. Here $\Lambda_0$ is the highest weight for $\cO$ and $\rho$ is the sum of the positive roots. This massive Klein-Gordon operator acts on sections and its kernel is invariant under the operator $D_{m,Z}$ given by $\frac{1}{i}$ times the covariant derivative in the $Z$ direction.

\begin{cor}{\label{corollary}}
	For $m$ sufficiently large operator $D_{m,Z}$ on $\ker\Box_m$ equipped with the energy form from Section \ref{vector_bundles} is self adjoint with discrete spectrum equal to the spectrum of $D_Z$ on $\cH_m$, and the multiplicity of $\lambda \in \sigma(D_Z)$ is the product of the multiplicity of $\lambda \in \sigma(D_{m,Z})$ and the dimension $d_m$ of the irreducible representation corresponding to $m\cO$. Thus if $\mu(E,\cV_m,\varphi)$ is defined for $D_{m,Z}$ in the same way $\mu(E,m,\varphi)$ is defined for $D_Z$ then under the clean intersection hypothesis we have
	\[ \mu(E,\cV_m,\varphi)\sim \frac{1}{d_m}\sum_{k=0}^{\infty} m^{n+\ell-1-k} b_k(\varphi,m). \]
\end{cor}
In general, one can compute the values of $\ell$ and $d_m$ in terms of the dominant integral element $\Lambda_0$. Indeed, if $R^+$ is the set of positive roots then
\[ \ell = \frac{1}{2}\dim\cO = \ \mbox{ the number of positive roots not orthogonal to } \Lambda_0 \]
and as a consequence of the Weyl character formula we have
\[ d_m = \frac{\prod_{\alpha\in R^+} \langle \alpha, \ m\Lambda_0 + \frac{1}{2}\rho\rangle}{ \prod_{\alpha\in R^+} \langle \alpha, \frac{1}{2}\rho\rangle}. \]
In particular we see that $d_m$ is a polynomial of degree $\ell$ and so our leading order asymptotics for $\mu(E,\cV_m,\varphi)$ as $m\to\infty$ are $m^{n-1}$. This is in agreement with \cite{strohmaier2020semi} where $\ell=0$ and $d_m=1$ for all $m$. When $G=\SU(2)$ and $\cO$ corresponds to the vector representation then $\ell=1$ and $d_m = m+1$. \\

We now provide an outline of the paper. In Section \ref{section_2} we demonstrate that $(M,g)$ being globally hyperbolic, stationary and spatially compact implies that this is also true for the Kaluza-Klein metric. The rest of the section is spent recalling the symplectic geometry from \cite{strohmaier2018gutzwiller} with $(M,g)$ replaced by $(P,g_{\omega})$. In Section \ref{section_2.1} we introduce the reduced phase space from \cite{guillemin1990reduction},\cite{weinstein1978universal},\cite{sternberg1977minimal} in the special case where the symplectic manifold is $\cN$ and in Section \ref{section_2.2} we study two aspects of periodic orbits on this reduced phase space: the linearized Poincar\'e first return map $P_{\gamma}$ and the holonomy map $\Hol_{\cO}$. In Section \ref{section_3} we use the general setup of \cite{strohmaier2018gutzwiller} to discuss the wave equation on $(P,g_{\omega})$. Since we allow for a potential (as long as it is constant along the fibers of $P$ and independent of $t$) the energy quadratic form on the space of solutions $\ker\Box_{\omega}$ need not be positive definite, but we use standard results from harmonic analysis together with some results on Krein and Pontryagin spaces to show that it is positive definite on the isotypic subspace $\cH_m$ for $m$ sufficiently large. In Section \ref{vector_bundles} we apply a result from \cite{islam2021gutzwiller} to obtain that $\mu(E,m,-)$ is tempered and then we provide a proof of Corollary \ref{corollary} given Theorems \ref{thm_1},\ref{thm_2},\ref{thm_3}. Finally, in Section \ref{section_4} we simply combine the techniques from \cite{strohmaier2018gutzwiller} and \cite{guillemin1990reduction} to obtain our main theorems.

\section{The Classical Dynamics: Wong's Equations}{\label{section_2}}
Let $(M^{n+1},g)$
be a connected, globally hyperbolic, oriented, time-oriented Lorentzian manifold.
For us, ``Lorentzian'' will mean that $g$ has signature
$(-1,+1,\cdots,+1)$. We refer the reader to \cite{o1983semi} for an
explanation of the various causality assumptions and related
terminology we will use. Throughout, we will assume the following:
\begin{enumerate}
	\item We will assume that $(M,g)$ admits a complete timelike
	Killing vector field $Z$, flowing forwards in time
	with respect to the time-orientation (and we will make
	a fixed choice of such a $Z$).
	\item We will assume that $(M,g)$ is spatially compact. That
	is, for some (hence any) choice of Cauchy hypersurface
	$\Sigma\subseteq M$, $\Sigma$ is compact.
\end{enumerate}

\begin{lem}{\cite{javaloyes2008note}} \\
	All such spacetimes $(M,g)$ as described above are diffeomorphic to $\bR_t\times \Sigma$ with metric of the form:
	\begin{align*}
		g &= -(N^2 -|\eta|_h^2)dt^2 +dt\otimes\eta +\eta\otimes dt +h \\
		&= -N^2 dt^2 +h_{ij}(dx^i +\beta^i dt)(dx^j +\beta^j dt)
	\end{align*}
	where $N:\Sigma \to \bR_{>0}$ is smooth, $\eta$ a 1-form on $\Sigma$,
	$\beta=\beta^i\partial_i$ a vector field on $\Sigma$, $h$ a Riemannian
	metric on $\Sigma$, and $\Sigma =
	\{0\}\times \Sigma$ a Cauchy hypersurface. We also have $N^2 > |\eta|^2_h$ pointwise and $\beta$ is the vector field $h$-dual to $\eta$. Furthermore, $\Sigma_t =
	\{t\}\times \Sigma$ will be a Cauchy hypersurface for each $t\in \bR$.
\end{lem}

Notice that such spacetimes $(M,g)$ are not necessarily static since
$Z=\partial_t$ need not be normal to $\Sigma$. Indeed, $N^{-1}(\partial_t - \beta)$ is the unit normal. Let's denote this by
\[ \nu:= N^{-1}(\partial_t -\beta)\]
noticing that this is a globally defined vector field on $M$ and is the unit normal to each $\Sigma_t$ when restricted to that submanifold. \\

The classical dynamics we wish to study are Wong's equation \cite{wong1970field} on a curved spacetime. In \cite{wong1970field}, the classical limit of a massive quantum particle in an external classical Yang-Mills field was determined to be given by the equations
\[ m\ddot{x}^i = \Tr(qF^{ij}_A)\dot{x}_j, \ \ \ \dot{x}^2 =-1 \]
where $q\in \fu(k)$ is a conserved quantity describing the internal degrees of freedom of the system (a generalization of charge). This is an analogue of the Lorentz force law, generalized to connections with structure groups $G$ other than $\Un(1)$. Some references for the study of these equations on curved non-relativistic space are \cite{sternberg1977minimal}, \cite{weinstein1978universal}. The basic idea is to express these equations as geodesic equations on a principal bundle over space equipped with a Kaluza-Klein metric. The Lorentzian analogue of this is developed below. \\

Now, let $G\subseteq \SO(k)$ be a compact Lie group and $\pi:P\to M$ a
principal $G$-bundle.

\begin{defn}
	Recalling that $(X,Y)\mapsto -\Tr(XY)$ is a
	positive definite $\Ad$-invariant inner product on $\fg$, given any
	connection $\omega$ on $P$ we obtain an induced \textit{Kaluza-Klein
		metric}:
	\[ g_{\omega} = \pi^* g - \Tr(\omega(-)\omega(-))\]
	on the total space of $P$. This is again a Lorentzian metric of
	signature $(-1,+1,\cdots,+1)$. Furthermore, we endow $(P,g_{\omega})$
	with the orientation induced from $M$ and $\fg\subseteq \fso(k)$, and
	the time-orientation induced from $M$. We let
	\[ Z^{\omega}:= \mbox{ the horizontal lift of } Z, \ \mbox{ and } \hat{\xi}:= \mbox{ the vertical vector field of } \xi\in \fg. \]
	We also abuse notation and set
	\[ \hat{n} := \mbox{ the horizontal lift of the unit normal } N^{-1}(\partial_t-\beta)=N^{-1}(Z-\beta).  \]
\end{defn}

\begin{lem}
	The horizontal lift $Z^{\omega}$ of $Z$ is a complete timelike
	future-oriented Killing vector field for $g_{\omega}$ and:
	\[ [Z^{\omega},\hat{\xi}] =0 \mbox{ for all } \xi\in \fg.\]
\end{lem}
\begin{proof}
	Recall that the connection $\omega:TP\to \fg$ defines a horizontal bundle $HP:= \ker\omega$ splitting $TP \to VP\cong P\times \fg$ and that $\pi_*|_{HP}:HP \to \pi^* TM$ is an isomorphism. The horizontal lift $Z^{\omega}$ of $Z$ is then given by
	\[ \left(\pi_*|_{HP}\right)^{-1}(\pi^*Z) \]
	where $\pi^*Z$ is the pullback of $Z$ to a section of $\pi^*TM$. Furthermore, the flow of $Z^{\omega}$ is the horizontal lift of the flow of $Z$ on $M$ (\cite{bleecker2005gauge} section 10.1). This immediately implies $Z^{\omega}$ is complete since $Z$ is. Furthermore since $Z$ is Killing on $(M,g)$ it follows that $\pi^*g$ is invariant under the flow of $Z^{\omega}$ and since $Z^{\omega}$ is horizontal we have
	\[ \cL_{Z^{\omega}}\Tr(\omega(-)\omega(-)^T) =0. \]
	Thus indeed $Z^{\omega}$ is Killing on $(P,g_{\omega})$. Finally, from \cite{bleecker2005gauge} section 2.2 we know that $Z^{\omega}$ is invariant under push-forward along the $G$-action on $P$ and therefore $[Z^{\omega},\hat{\xi}]=0$ for all $\xi\in\fg$.
\end{proof}

The next result is an immediate corollary of \cite{sanchez2007some} but we include the proof from \cite{sanchez2007some} for the reader's convenience. Note that we are using the intrinsic definition of Cauchy hypersurfaces in terms of inextendible causal curves since this is the definition used when proving well-posedness for the wave equation in \cite{bar2007wave} and we would like to apply their results directly. However, the proof of well-posedness simplifies when the manifold is spatially compact with a complete timelike Killing vector field, and so our discussion of ``inextendibility'' below is unnecessary for readers working exclusively in this setting. \\

For a particularly nice discussion of the geometry of these spacetimes we refer the reader to \cite{anderson2000stationary} and \cite{MR3428355} where it is also shown that when $n=3$ these spacetimes do not admit non-trivial (i.e. product with a flat Riemannian spacetime) vacuum solutions to the Einstein equations.

\begin{lem}{\cite{sanchez2007some}\label{total_globally_hyp}} \\
	$(P,g_{\omega})$ is globally hyperbolic and each $P_t =
	\pi^{-1}(\Sigma_t)$, $t\in \bR$, a Cauchy hypersurface.
\end{lem}
\begin{proof}
	The map $\bR\times P_0\to P$ induced by flowing along
	$Z^{\omega}$ is a global diffeomorphism and so we can, without
	loss of generality, assume $P=\bR_t\times P_0$ as a manifold
	with $Z^{\omega}=\partial_t$ and $P_t = \{t\}\times P_0$.
	Using our standard form for the metric $g$ on $M$ we can write
	the Kaluza-Klein metric on $P$ as:
	\begin{equation}{\label{metric_form}} g_{\omega} = -( (N\circ \pi)^2 -|\pi^*\eta|^2_{\pi^* h})
	dt^2 +dt\otimes (\pi^*\eta) +(\pi^*\eta)\otimes dt +
	(\pi^* h - \Tr(\omega(-)\omega(-))).\end{equation}
	In particular, the Riemannian metric $\tilde{h}_t$ on $P_t$ induced by pulling back $g_{\omega}$ is independent of $t$:
	\[ \tilde{h}_t=\tilde{h} = \pi^*h - \Tr(\omega(-)\omega(-)). \]
	
	Choose now an arbitrary inextendible causal curve
	$\gamma$ in $P$, parametrized with respect to $t$ so that
	$\gamma(t)=(t,\gamma_0(t))$. We can always make this
	parametrization for causal curves in the spacetimes we are
	considering thanks to the global time function $t$.  We now write $(a,b)\subseteq \bR\cup
	\{\pm\infty\}$ for
	the domain of $\gamma$ with $b \in (-\infty,\infty]$ and $a\in
	[-\infty,b)$. Since $\gamma$ is causal we have at all $t\in
	(a,b)$ that:
	\[ -N(\pi(\gamma_0(t)))^2 +|\pi^*\eta|^2_{\pi^* h}(
	\gamma(t)) +2\langle \pi^*\eta,\gamma_0'\rangle_{\tilde{h}} +\tilde{h}(
	\gamma_0',\gamma_0') \leq 0.\]
	Thus by Cauchy-Schwarz we have:
	\[  |\gamma_0'|^2_{\tilde{h}} -2|\pi^*\eta|_{\tilde{h}}
	|\gamma_0'|_{\tilde{h}}-N^2
	+|\pi^*\eta|_{\tilde{h}}^2 \leq 0.\]
	Rearranging yields:
	\[ (|\gamma_0'|_{\tilde{h}} - |\pi^*\eta|_{\tilde{h}})^2 \leq
	N^2.\]
	Now, suppose for contradiction that $b<\infty$ (the case
	$-\infty <a$ is completely analogous) and set:
	\[ C:= \sup_{[-|b|-1,b+1]\times P_0} \left(
	N +|\pi^*\eta|_{\tilde{h}}\right).\]
	We have $C<\infty$ since $\Sigma$ and $G$ are compact, hence
	$P_0$ and $[-|b|-1,b+1]\times P_0$ are compact. We then have:
	\[ |\gamma_0'|_{\tilde{h}}\leq C \mbox{ on } [-|b|-1,b]\]
	and so the curve $\gamma$ must be extendible beyond time $b$,
	a contradiction.
\end{proof}

Here we provide a brief remark on the above proof: recall that in the
definition of a Cauchy hypersurface, the assumption is that all
curves which are both inextendible and causal intersect the
hypersurface exactly once. One does not make this requirement
of causal curves which are extendible, but whose extensions
are non-causal. To see why, consider the following example on flat 2-dimensional Minkowski space (this example also works in any dimension). \\

We let $t,x$ denote our coordinates so that the metric on $\bR^2$ is $-dt^2 +dx^2$. Consider the curve given by $t(s)=s$ and $x(s) = 0$ for $s\leq 0$ and $x(s) = e^{-1/s^2}$ for $s>0$. There exists a time $s_0 >0$ so that $\gamma(s):=(t(s),x(s))$ is a causal curve for $s\in (-\infty,s_0)$ but there exists no $\epsilon >0$ on which $\gamma$ extends to a causal curve on $(-\infty,s_0+\epsilon)$. So $\gamma$ has no causal extensions, but it is not inextendible since it does admit an extension to a curve $\gamma:\bR\to \bR^2$ (albeit a non-causal one). This clarifies why, in the definition of a Cauchy hypersurface, one only requires all inextendible curves, which are also causal, to intersect the Cauchy hypersurface exactly once. \\

As an immediate corollary of the specific form of the metric $g_{\omega}$ derived in the above proof, we also obtain the following.

\begin{cor}
	The horizontal lift $\hat{n}$ of the unit normal to the Cauchy hypersurfaces $\Sigma_t\subseteq M$ is the unit normal to the Cauchy hypersurfaces $P_t\subseteq P$ with respect to $g_{\omega}$.
\end{cor}

Let's now discuss the geodesic equations in the Kaluza-Klein spacetime $(P,g_{\omega})$.  Most of the basic facts here can be found in the reference \cite{bleecker2005gauge} but we include them for the reader's convenience together with section and/or theorem numbers from \cite{bleecker2005gauge}. Recall that the Levi-Civita connection, together with the ODE $\nabla_{\dot{\gamma}}\dot{\gamma}=0$ defining the geodesic equations are defined on Lorentzian manifolds in the exact same way as they are defined on Riemannian manifolds. Furthermore, we still have
\[ \frac{d}{ds}g_{\omega}(\dot{\gamma}(s),\dot{\gamma}(s)) = 2g_{\omega}(\nabla_{\dot{\gamma}(s)}\dot{\gamma}(s),\dot{\gamma}(s))=0 \]
and so
\[ g_{\omega}(\dot{\gamma},\dot{\gamma}) \mbox{ is constant along geoedesics.} \]
This allows us to split geodesics into three types.
\begin{defn}
	We call a geodesic $\gamma$ on $(P,g_{\omega})$ \textbf{lightlike} (respectively \textbf{spacelike} and \textbf{null}) if and only if the constant $g_{\omega}(\dot{\gamma},\dot{\gamma})$ is negative (respectively positive and zero).
\end{defn}
As the below lemma explains, despite Wong's equations describing massive particles in $(M,g)$, we will be interested in null geodesics in $(P,g_{\omega})$.

\begin{lem}
	Let $\gamma$ be a geodesic in $(P,g_{\omega})$. Then the value $\omega(\gamma'(t))$ is constant. Thus
	\[ g_{\omega}(\dot{\gamma},\dot{\gamma}) = g( (\pi\circ\gamma)', (\pi\circ\gamma)' ) - \Tr(\omega(\dot{\gamma}),\omega(\dot{\gamma})) \]
	being constant implies that the projected curve $\pi\circ\gamma$ in $M$ has $g( (\pi\circ\gamma)',(\pi\circ\gamma)')$ constant. Since $-\Tr(\omega(\dot{\gamma}),\omega(\dot{\gamma})) \geq 0$ (and is zero if and only if $\omega(\dot{\gamma})=0$) we see that:
	\[ g_{\omega}(\dot{\gamma},\dot{\gamma})\leq 0 \ \mbox{ implies } \ g( (\pi\circ\gamma)', (\pi\circ\gamma)' )\leq 0 \]
	and so timelike or null geodesics in $(P,g_{\omega})$ project to timelike or null curves in $(M,g)$. In fact, the projection will be timelike unless the geodesic in $P$ is null and $\omega(\dot{\gamma})=0$.
\end{lem}
\begin{proof}
	The only part of this not proven in the statement is that $\omega(\dot{\gamma}(t))$ is constant for a geodesic $\gamma$ in $(P,g_{\omega})$. This can be found in \cite{bleecker2005gauge} theorem 10.1.5.
\end{proof}

\begin{lem}
	In the special case where $M=\bR_t\times \bR^n$ is flat and $P=M\times G$, null geodesics in $P$ project to solutions to Wong's equations with $\dot{x}^2$ a non-positive constant. More generally, null geodesics in $(P,g_{\omega})$ project to curves $\gamma$ in $(M,g)$ together with a section $q$ of $\gamma^*\Ad(P)$ satisfying:
	\[ \begin{cases} \nabla_{\dot{\gamma}}\dot{\gamma} &= -\dot{\gamma}\llcorner \Tr(qF_A) \\ (\gamma^*\nabla^A)q &=0 \\ g(\dot{\gamma},\dot{\gamma}) &= \mbox{ constant} \end{cases} \]
	where $A$ is the connection on the bundle $\Ad(P)$ induced by $\omega$.
\end{lem}
\begin{proof}
	Let $\tilde{\gamma}$ be a null geodesic in $P$ and $\gamma$ the projected curve in $M$. We denote $q:=\omega(\dot{\tilde{\gamma}}(t))$ and notice that by $\Ad$-equivarance of connection 1-forms on principle bundles this defines a section of $\gamma^*\Ad(P)$. We've already seen that $g(\dot{\gamma},\dot{\gamma})$ is constant so it suffices to prove that $q$ is covariantly constant with respect to $\nabla^A$ and that the geodesic equations reduce to Wong's equations. \\
	
	The fact that the geodesic equations in $(P,g_{\omega})$ reduce to Wong's equations on $M$ is theorem 10.1.6 in \cite{bleecker2005gauge}. As for $q$, we note that its covariant derivative as a section of $\gamma^*\Ad(P)$ is just the horizontal part of its time derivative as a $\fg$-valued function on a curve in $P$, and this is zero since the entire time derivative vanishes.
\end{proof}

Given a null geodesic $\gamma$ in $(P,g_{\omega})$, we would like to think of the constant $\omega(\dot{\gamma})$ as the ``charge''. Unfortunately, unlike the abelian case of the Lorentz force law, different lifts of solutions to Wong's equations in $M$ to geodesics in $(P,g_{\omega})$ will have different charges. Indeed, if the two lifts of our curve in $M$ are related by the right action of $g\in G$ on $P$ then the charges of the two lifts will be related by $\Ad_g$. Identifying the charge $q$ with $-\Tr(q(-)) \in \fg^*$ we arrive at the following gauge invariant definition of charge.

\begin{defn}
	Let $\gamma$ be a null geodesic in $(P,g_{\omega})$ and $\xi_0 := -\Tr(\omega(\dot{\gamma})(-)) \in \fg^*$. The \textbf{charge} of $\gamma$ is defined to be the coadjoint orbit:
	\[ \cO:= \{ \Ad_g^*\xi_0 \ : \ g\in G \}\subseteq \fg^*. \]
\end{defn}

Just as in the flat case, Wong's equations on a curved spacetime will arise as classical limits of the quantum system. One consequence of this will be \textit{charge quantization}. \\

For now, let's proceed to the Hamiltonian description of the dynamics of these null geodesics. Recall that the relativistic description of the phase space of a system is simply the space of solutions to the equations of motion, and the identification with a cotangent bundle arises from the equations typically being second order ODE and so solutions correspond to initial data. In this way, the following results and definitions can been seen as relativistic versions of the results on the phase space for Wong's equations from \cite{sternberg1977minimal},\cite{weinstein1978universal}.

\begin{defn}
	The \textbf{null bicharacteristic flow} $G_s$ is the Hamiltonian flow on $T^*P\setminus 0$ of the Hamiltonian $\xi\mapsto \frac{1}{2}g_{\omega}^{-1}(\xi,\xi)$.
\end{defn}

\begin{lem}
	Let $\Phi^Z_s$ and $\Phi^{\xi}_s$ respectively denote the flows on $T^*P\setminus 0$ given by the derivatives of the flows of $Z^{\omega}$ and $\hat{\xi}$ ($\xi\in\fg$) on $P$. Then $\Phi_s^Z$ and $\Phi^{\xi}_s$ commute with $G_s$ for every $\xi\in\fg$.
\end{lem}
\begin{proof}
	Since $\Phi^Z_s$ and $\Phi^{\xi}_s$ are derivatives of flows on $P$ they are a 1-parameter family of canonical transformations on $T^*P\setminus 0$ and therefore, by the Hamiltonian version of Noether's theorem, it suffices to show that the Hamiltonian $\xi \mapsto \frac{1}{2}g_{\omega}^{-1}(\xi,\xi)$ is invariant under the flows $\Phi^Z_s,\Phi^{\xi}_s$ in order to prove that they commute with $G_s$. But this is immediate from both $Z^{\omega}$ and $\hat{\xi}$ being Killing vector fields for the metric $g_{\omega}$.
\end{proof}

One incredibly important subtlety is the following. Since our spacetimes are not necessarily ultrastatic, there is no reason to expect that if $\eta\in T^*P_0$ then $G_s(\eta)\in T^*P_s$. This is our reason for using the variable $s$ instead of $t$. Indeed, we do have (by definition) that $\Phi^Z_t(\eta) \in T^*P_t$ for $\eta\in T^*P_0$.

\begin{defn}
	We begin by denoting:
	\[ \cN_a := \{ \mbox{ all future-directed, inextendible null geodesics in } (P,g_{\omega}) \}. \]
	Recall that, for us, geodesics are specifically solutions to the equation $\nabla_{\dot{\gamma}}\dot{\gamma}=0$ and hence these are already ``affinely parametrized''. By definition, elements of $\cN_a$ are, in particular, inextendible causal curves and therefore intersect each Cauchy hypersurface $P_t$ exactly once. An invariant way of dealing with the fact that such curves $\gamma$ need not satisfy $\gamma(0)\in P_0$ is to define
	\[ \cN:=\cN_a/\bR \]
	where $b\in \bR$ acts on $\cN_a$ by $\gamma(s)\mapsto \gamma(s+b)$. Notice additionally that the $\bR_{>0}$-action on $\cN_a$ where $a\in\bR_{>0}$ acts by $\gamma(s)\mapsto \gamma(as)$ descends to an $\bR_{>0}$-action on the quotient $\cN$.
\end{defn}

The above set $\cN$ with $\bR_{>0}$-action is naturally a symplectic manifold with $\bR_{>0}$-action and one can view the next lemma as saying the there are $\bR_{>0}$-equivariant Cauchy-data symplectomorphisms between $\cN$ and cotangent bundles. Instead, we will simply take the next lemma as a definition of the smooth manifold and symplectic structures. For this, we will need the definition.

\begin{defn}
	We define three sub-cone-bundles of $T^*P\setminus 0$:
	\[ T^*_0P := \{ \zeta\in T^*P\setminus 0 \ : \ g^{-1}_{\omega}(\zeta,\zeta)=0 \} \]
	and
	\[ T^*_{\pm}P:= \{\zeta\in T^*_0P \ : \ \zeta \mbox{ is future (respectively past) oriented}\} \]
	so that $T^*_0P = T_+^*P \sqcup T_-^*P$.
\end{defn}

As in \cite{strohmaier2018gutzwiller} there are natural isomorphisms of bundles over $P_0$:
\[ T^*P_0\setminus 0\cong T^*_+P|_{P_0} \ \mbox{ and } \ T^*P_0\setminus 0\cong T^*_-P|_{P_0} \]
which are symplectomorphisms but are not $\bR_{>0}$-equivariant! Indeed, the map
\[ T^*P_0 \setminus 0 \xrightarrow{\cong} T^*_+P|_{P_0} \]
is given by $\zeta \mapsto \zeta +|\zeta|_h^2\hat{n}$.

\begin{lem}{\label{lem2.12}}
	Each equivalence class in $\cN$ has a unique representative $\gamma:\bR\to P$ satisfying $\gamma(0)\in P_0$. Identifying elements of $\cN$ with these representatives gives us an $\bR_{>0}$-equivariant bijection
	\begin{align*}
		\cN &\xrightarrow{\cong} T^*_+P|_{P_0} \\
		\gamma &\mapsto (\gamma(0),\dot{\gamma}(0))
	\end{align*}
	where $\gamma'(0)$ is identified with a cotangent vector via $g_{\omega}$. The $\bR_{>0}$-action on $T^*_+P|_{P_0}$ is given by scalar multiplication in the fibers. Furthermore, the inverse of the above bijection is given by
	\[ T^*_+P|_{P_0}\setminus 0 \ni \eta \mapsto (s\mapsto G_s(\eta)) \]
	or, more precisely, $\eta$ maps to the projection of the curve $s\mapsto G_s(\eta)$ down to $P$.
\end{lem}
\begin{proof}
	This is immediate from the definition of a future-directed, inextendible null geodesic and the existence and uniqueness of solutions to ODE.
\end{proof}

\begin{lem}
	$g\in G$ has a right action on $\cN$ induced by its right action on $\cN_a$ given by $\gamma(s)\mapsto \gamma(s)g$ (via the right action on $P$). The bijection in \ref{lem2.12} intertwines this right action with the right action on $T^*_+P|_{P_0}$ given by dualizing (using $g_{\omega}$) the action of pushing forward by right multiplication by $g$ on $P$.
\end{lem}
\begin{proof}
	This is immediate from the explicit form of our isomorphism $\cN\cong T^*_+P|_{P_0}$ and the fact that $G$ acts by isometries and therefore leaves $T^*_+P|_{P_0}$ invariant.
\end{proof}

\begin{lem}
	The flows $\Phi^Z_s$ and $\Phi^{\xi}_s$ on $\cN$ induced by \ref{lem2.12} are Hamiltonian flows with respective Hamiltonians:
	\[ H_Z(\gamma) = Z^{\omega}\llcorner (\gamma(0),\dot{\gamma}(0)) \ \mbox{ and } \ H_{\xi}(\gamma)= \hat{\xi}\llcorner (\gamma(0),\dot{\gamma}(0)) \]
	where again we have chosen representative geodesics $\gamma$ with $\gamma(0)\in P_0$.
	Furthermore, the $\Phi^{\xi}_s$'s arise (through the exponential map) from the natural right-action of $G$ on $\cN$ hence this $G$-action is Hamiltonian.
\end{lem}

As the above right $G$-action is Hamiltonian, we can consider its \textbf{moment-map}:
\begin{align*}
	\mu:\cN &\to \fg^* \\
	\langle \mu(\gamma),\xi\rangle &= H_{\xi}(\gamma).
\end{align*}

\begin{lem}
	Under the isomorphism $\fg\cong \fg^*$ induced by our $\Ad$-invariant inner product on $\fg$, the moment map is given by
	\[ \gamma \mapsto \omega(\dot{\gamma}). \]
\end{lem}
\begin{proof}
	We know that $\omega(\hat{\xi})=\xi$ by the definition of a connection on a principal bundle and so the result follows from
	\[ \hat{\xi}\llcorner (\gamma(0),\dot{\gamma}(0)) = \Tr(\omega(\hat{\xi})\omega(\dot{\gamma})^T) \]
	since we're using $g_{\omega}$ to identify $\dot{\gamma}(0)$ with a covector.
\end{proof}

As a final remark before we proceed to symplectic reduction, we demonstrate that, while our Hamiltonian may appear linear (indeed, it is homogeneous of degree 1 with respect to the $\bR_{>0}$-action on $\cN$), it is in fact quadratic after applying the symplectomorphism $T^*P_0\setminus 0\cong \cN$.

\begin{lem}
	Under the symplectomorphism $T^*P_0\setminus 0\cong \cN$ the Hamiltonian $H_Z$ becomes:
	\begin{align*}
		H_Z:T^*P_0\setminus 0 &\to \bR \\
		\zeta &\mapsto N|\zeta|_h^2 + \langle \eta,\zeta\rangle_h
	\end{align*}
	where $\eta$ is the 1-form on $P_0$ coming from our explicit form for the metric $g_{\omega}$ in \ref{metric_form}.
\end{lem}
\begin{proof}
	This follows from the isomorphism $T^*P_0\setminus \cong T^*_+P|_{P_0}$ being given by $\zeta \mapsto \zeta + |\zeta|_h^2 \hat{n}$ and $g_{\omega}(Z^{\omega},\hat{n}) = N^{-1}$.
\end{proof}

Notice in particular that the fact that $N>|\eta|_h$ pointwise implies that $H_Z$ is strictly positive. Furthermore, if we had $N-|\eta|_h$ uniformly bounded away from zero on $\Sigma_0$ then $H_Z$ would both be uniformly bounded away from zero and would have a uniformly positive definite fiberwise Hessian.

\subsection{The Reduced Phase Space}{\label{section_2.1}}

Fix a charge, i.e. a coadjoint orbit $\cO\subseteq \fg^*$. We now wish to form the symplectically reduced phase space of solutions with charge $\cO$. The construction of this in Riemannian signature, and its relationship to Wong's equations can be found in \cite{guillemin1990symplectic} and it generalizes with almost no modifications to our setting. \\

Recall that our coadjoint orbit $\cO$ is naturally a symplectic manifold. The symplectic form $\omega_{\cO}$ can be defined as follows. Fix $\xi_0\in \cO$ and let $G_{\xi_0}$ denote the stabilizer of $\xi_0$ under the coadjoint action. Then
\begin{align*}
	G &\to \cO \\
	g &\mapsto \Ad_g^*\xi_0
\end{align*}
induces an isomorphism
\[ G/G_{\xi_0}\cong \cO \]
which identifies
\[ T_{\xi_0}\cO \cong \fg/\fg_{\xi_0} \]
where $\fg_{\xi_0}$ is the Lie algebra of $G_{\xi_0}$. The other tangent spaces of $\cO$ are also identified with $\fg/\fg_{\xi_0}$ by pushforward along the $G$-action. We then have:
\[ \omega_{\cO}(X,Y) = \langle \xi_0,[X,Y]\rangle. \]
We notice that this is well-defined on $\fg/\fg_{\xi_0}$ since
\[ \fg_{\xi_0} = \{X\in \fg \ : \ \langle \xi_0,[X,Y]\rangle =0 \ \mbox{ for all } Y\in \fg\}. \]
Let $\bar{\cO}$ denote $\cO$ but equipped with $-\omega_{\cO}$ as its symplectic form instead of $\cO$.

\begin{lem}
	The extended moment map
	\begin{align*}
		\mu_{\cO}: \cN\times \bar{\cO} &\to \fg^* \\
		\mu_{\cO}(\gamma,\xi) &:= \mu(\gamma)-\xi
	\end{align*}
	is a submersion and $G$ acts freely on $\mu_{\cO}^{-1}(0)$.
\end{lem}
\begin{proof}
	The fact that $G$ acts freely on $\mu_{\cO}^{-1}(0)$ simply follows from $G$ acting freely on $\cN\cong T^*_+P|_{P_0}$ since $P,P_0$ are principal $G$-bundles. To see that $\mu_{\cO}$ is a submersion, we notice that under the isomorphism $\cN\cong T^*_+P|_{P_0}$ we have
	\begin{align*}
		\mu_{\cO}:T^*_+P|_{P_0} \times \bar{\cO} &\to \fg^* \\
		(\zeta,\xi) &\mapsto \Tr(\omega(\zeta)^T\omega(-))-\xi
	\end{align*}
	and if we use our $\Ad$-invariant inner product to identify $\fg\cong \fg^*$ then this maps
	\[ (\zeta,\xi) \mapsto \omega(\zeta)-\xi. \]
	Forgetting $\xi$ we can already see that $\zeta\mapsto \omega(\zeta)$ is a submersion (and therefore $\mu_{\cO}$ is a submersion). Indeed, it suffices to prove that for every $\xi\in\fg$ there exists $\zeta\in T^*_+P|_{P_0}$ such that $\omega(\zeta)=\xi$. However, $\hat{\xi}$ is tangent to $P_0$ with $g_{\omega}(\hat{\xi},\hat{\xi})=\Tr(\xi\xi^T)$ so $\zeta := \Tr(\xi\xi^T)\hat{n}+\hat{\xi}$ is future-directed, has $\omega(\zeta)=\xi$ and $g_{\omega}(\zeta,\zeta) =0$ as desired.
\end{proof}

From the above proof we record as a remark the fact that $\mu_{\cO}^{-1}(0)$ is precisely the space of pairs $(\gamma,\xi)$ where $\gamma\in\cN$ and $\xi\in\cO$ satisfy
\[ \Tr(\omega(\dot{\gamma})^T(-)) = \xi. \]
This $\mu_{\cO}^{-1}(0)$ is precisely the space of solutions with charge $\cO$, prior to quotienting by gauge transformations.

\begin{defn}
	The \textbf{reduced phase space} is
	\[ \cN_{\cO}:= \mu_{\cO}^{-1}(0)/G \]
	with symplectic form obtained from the one on $\cN\times \bar{\cO}$.
\end{defn}

\begin{lem}
	The Hamiltonian $H_Z$, extended to $\cN\times \bar{\cO}$ to be independent of $\bar{\cO}$, is invariant under the $G$-action and therefore descends to a Hamiltonian $\tilde{H}_Z$ on $\cN_{\cO}$ with flow $\tilde{\Phi}^Z_s$.
\end{lem}
\begin{proof}
	From the definition of $H_Z$ we see that what we have to show is that $g_{\omega}(Z^{\omega},\dot{\gamma}\cdot g) = g_{\omega}(Z^{\omega},\dot{\gamma})$ for all $g\in G$ and $\gamma\in \cN$. However:
	\[ g_{\omega}(Z^{\omega},\dot{\gamma}\cdot g) = g_{\omega}(Z^{\omega}\cdot g^{-1},\dot{\gamma}) = g_{\omega}(Z^{\omega},\dot{\gamma}) \]
	since $Z^{\omega}=\partial_t$ is invariant under the $G$ action.
\end{proof}

The point of the previous construction is its manifestly gauge-invariant nature. Below we give an alternative characterization that might be more familiar to some readers, although we will not use it in our proof. \\

Fix $\xi_0\in \cO$ and recall from our proof that $\mu_{\cO}$ is a submersion that $\mu$ is also a submersion, hence $\xi_0$ is automatically a regular value. Furthermore, while the full $G$-action on $\cN$ doesn't preserve the submanifold $\mu^{-1}(\xi_0)$, it is preserved by the action of the stabilizer $G_{\xi_0}$ of $\xi_0$. The action of $G_{\xi_0}$ on $\mu^{-1}(\xi_0)$ is free since the action of $G$ on $\cN$ is free.

\begin{defn}
	The \textbf{reduced phase space} (version II) is the quotient
	\[ \mu^{-1}(\xi_0)/G_{\xi_0} \]
	with the symplectic form induced from that on $\cN$.
\end{defn}

\begin{lem}{\cite{guillemin1990reduction}}
	The map
	\begin{align*}
		\mu^{-1}(\xi_0) &\to \cN_{\cO} \\
		\gamma &\mapsto [(\gamma, \xi_0)] 
	\end{align*}
	induces a symplectomorphism $\mu^{-1}(\xi_0)/G_{\xi_0}\cong \cN_{\cO}$ intertwining the reductions of the Hamiltonian flow of $H_Z$ to $\mu^{-1}(\xi_0)/G_{\xi_0}$ and $\cN_{\cO}$. Here $[(\gamma,\xi_0)]$ denotes the equivalence class of $(\gamma,\xi_0)$ in the quotient.
\end{lem}

\subsection{Periodic Orbits}{\label{section_2.2}}

Finally, let's note that since $M$ is assumed to be spatially compact we expect the quantum system to have discrete spectrum and hence bound states. The leading order singularities in our distributional trace of the propagator will be therefore expressed as a sum over classical bound states: periodic orbits of null geodesics under $\tilde{\Phi}^Z_s$. There are two aspects of these  periodic orbits we will need to consider:
\begin{enumerate}
	\item the \textit{(linearized) Poincar\'e first return map} of a periodic orbit, and
	\item the phase change due to a periodic orbit for the Aharonov-Bohm effect.
\end{enumerate}
The first of these points relates to the classical dynamics of periodic orbits, while the second of these is only relevant for the quantum effects we will discuss later. \\

Following \cite{strohmaier2018gutzwiller}, we fix an energy $E\in\bR$ and restrict ourselves to the contact manifold given by the level surface
\[ \tilde{H}_Z^{-1}(E)\subseteq \cN_{\cO}. \]
This is invariant under the $\tilde{\Phi}^Z_s$-flow and so we can define the set of \textbf{periods}:
\[ \cP_E:= \{T\in \bR\setminus \{0\} \ : \ \exists z\in \tilde{H}_Z^{-1}(E) \ \mbox{ such that } \ \tilde{\Phi}^Z_T(z)=z\} \]
and, for $T\in\cP_E$, the set of \textbf{periodic points}:
\[ \cP_{E,T}:= \{ z\in \tilde{H}_Z^{-1}(E) \ : \ \tilde{\Phi}_T^Z(z)=z. \} \]
We say that $T>0$ is the \textbf{minimum period} of $z$ if and only if it is the smallest positive time for which $\tilde{\Phi}_T^Z(z)=z$. The below result is a general fact concerning Hamiltonian dynamics and is a simple consequence of the implicit function theorem.

\begin{lem}{(\cite{easton1993introduction} Prop 8.5.3)} \\
	Given a periodic point $z_0\in \cP_{E,T}$ where $T$ is its minimum period there exists, in a sufficiently small neighborhood of $z_0$, a codimension 1 symplectic submanifold
	\[ z_0\in S\subseteq \tilde{H}_Z^{-1}(E) \]
	which is transverse to the flow $\tilde{\Phi}_s^Z$. Furthermore, in a sufficiently small neighborhood of $z_0$ in $S$, the \textbf{first return time}
	\[ \cT(z) := \min\{t>0 \ : \ \tilde{\Phi}_t^Z(z)\in S\} \]
	is well-defined, smooth and satisfies $\cT(z_0)=T$.
\end{lem}

\begin{defn}
	With $z_0,S,T$ as above, we define the \textbf{linearized Poincar\'e first return map} to be
	\[ P_{z_0,S}:= \frac{\partial}{\partial z}\Big|_{z=z_0} \tilde{\Phi}^Z_{\cT(z)}(z):T_{z_0}S\to T_{z_0}S. \]
	This is a linear symplectic map. For any other choice of local symplectic transversal $S'$ there is a linear symplectic isomorphism
	\[ L:T_{z_0}S' \xrightarrow{\cong} T_{z_0}S \]
	such that
	\[ P_{z_0,S'} = L^{-1}\circ P_{z_0,S}\circ L. \]
\end{defn}

There is actually an alternate, perhaps simpler, description of these maps $P_{z_0,S}$. This alternate description is analogous to the more standard definition of the linearized Poincar\'e first return map for geodesic flow on Riemannian or Lorentzian manifolds, which is usually defined with the aid of Jacobi fields.

\begin{defn}
	Given $z_0\in \cP_{E,T}$ with $T$ the minimum period of $z_0$, we define the \textbf{Floquet operator} of $z_0$ to be:
	\[ V_{z_0}(T):= \frac{d}{dz}\Big|_{z=z_0} \tilde{\Phi}^Z_T(z):T_{z_0}\cN_{\cO}\to T_{z_0}\cN_{\cO}. \]
\end{defn}

\begin{lem}{\label{floquet}}
	The subspace
	\[ W_{z_0}:= \Span\{ \tilde{Z}(z_0), \ \nabla \tilde{H}(z_0) \} \]
	is symplectic, as is the quotient $T_{z_0}\cN_{\cO}/W_{z_0}$, and $W_{z_0}$ is preserved by the Floquet operator. The induced quotient map
	\[ V_{z_0}(T): T_{z_0}\cN_{\cO}/W_{z_0}\to T_{z_0}\cN_{\cO}/W_{z_0} \]
	is conjugate via a linear symplectomorphism to the linearized Poincar\'e first return map.
\end{lem}

Let's discuss for some time the significance of these operators to us. For this, we will need the following assumption.

\begin{defn}
	We say that $E$ satisfies the \textbf{clean intersection hypothesis} if and only if $E$ is a regular value for $\tilde{H}_Z$ and the flow map
	\begin{align*}
		\bR\times \tilde{H}_Z^{-1}(E) &\to \tilde{H}_Z^{-1}(E)\times \tilde{H}_Z^{-1}(E) \\
		(t,\gamma) &\mapsto (\gamma, \tilde{\Phi}^Z_t(\gamma))
	\end{align*}
	admits a clean fibered product over $\tilde{H}^{-1}(E)\times \tilde{H}^{-1}(E)$ with the diagonal map $\tilde{H}^{-1}(E)\to \tilde{H}^{-1}(E)\times \tilde{H}^{-1}(E)$.
\end{defn}

Let's discuss this hypothesis for a moment. The fibered product is given, as a set, by:
\[ \mathfrak{Y}_E:=\{ (T,\gamma)\in \bR\times \tilde{H}^{-1}_Z(E) \ : \ \tilde{\Phi}^Z_T(\gamma)=\gamma \}. \]
Notice that this contains $\{0\}\times \tilde{H}_Z^{-1}(E)$ as a subset and the clean intersection hypothesis implies that $\mathfrak{Y}_E$ is a disjoint union of smooth submanifolds of $\bR\times \tilde{H}^{-1}_Z(E)$.

\begin{lem}{\label{set_of_orbits}}
	Under the clean intersection hypothesis, $\{0\}\times \tilde{H}_Z^{-1}(E)$ is a clopen subset of $\fY_E$ and every connected component $Y\subseteq \fY_E$ has
	\[ \dim(Y)\leq \dim \tilde{H}_Z^{-1}(E) = 2n+\dim\cO -1. \]
\end{lem}
\begin{proof}
	Let $Y\subseteq \fY_E$ be any connected component. By the clean intersection hypothesis, for any $(T,\gamma)\in Y$ we must have
	\[ T_{(T,\gamma)}Y = \left\{ (\tau, \zeta) \in T_T\bR\times T_{\gamma}\tilde{H}^{-1}_Z(E) \ : \ \tau\frac{d}{dt}\Big|_{t=T}\tilde{\Phi}^Z_t(\gamma)+ D\tilde{\Phi}^Z_T(\zeta) = \zeta \right\} \]
	Since $\zeta \mapsto \zeta - D\tilde{\Phi}^Z_T(\zeta)$ is linear the only way for the above constraint to be trivial (and not reduce the dimension) is if $ \frac{d}{dt}|_{t=T}\tilde{\Phi}^Z_t(\gamma)=0$ and if $D\tilde{\Phi}^Z_T = \id$. Indeed, if $\frac{d}{dt}|_{t=T}\tilde{\Phi}^Z_t(\gamma)\neq 0$ and we didn't want the equation to constrain $\zeta$ then we would need to constrain $\tau$ to $\tau=0$. But now since $\tilde{\Phi}^Z_T(\gamma)=\gamma$ it follows that $\frac{d}{dt}|_{t=T}\tilde{\Phi}^Z_t(\gamma)=0$ implies $\frac{d}{dt}|_{t=0}\tilde{\Phi}^Z_t(\gamma)=0$ and so the gradient of the Hamiltonian $\nabla \tilde{H}_Z$ vanishes at $\gamma$ and so $\gamma$ is an equilibrium point. However, we assumed that $\gamma \in \tilde{H}^{-1}_Z(E)$ and that $E$ was a regular value for $\tilde{H}_Z$, which contradicts $\nabla\tilde{H}_Z$ vanishing at $\gamma$. \\
	
	Now, let $Y$ be the smallest clopen subset containing $\{0\}\times \tilde{H}_Z^{-1}(E)$. We have already shown that $\dim(Y')\leq \dim\tilde{H}_Z^{-1}(E)$ for any connected component $Y'$ and so we must have $\dim(Y) = \dim\tilde{H}_Z^{-1}(E)$ since $Y$ is a disjoint union of connected components. In particular, since the inclusion
	\[ \{0\}\times \tilde{H}_Z^{-1}(E) \hookrightarrow Y \]
	is an immersion it is automatically a submersion as well and hence a local diffeomorphism. Local diffeomorphisms are local homeomorphisms and are hence open maps. Thus the image $\{0\}\times \tilde{H}_Z^{-1}(E)$ is open in $Y$, hence open in $\fY_E$ since $Y$ is open in $\fY_E$. Since $\{0\}\times \tilde{H}_Z^{-1}(E)$ is also closed in $\fY_E$ it follows that it is clopen hence
	\[ \{0\}\times \tilde{H}_Z^{-1}(E) = Y \]
	as desired.
\end{proof}

We should remark that there is no reason to expect $\tilde{H}_Z^{-1}(E)$ to be connected even if $M$ is connected since we have allowed disconnected structure groups such as $G=\Or(d)$. \\

In our trace formula, the leading order singularities of the distributional trace will have symbols given by integrals over components of the above clean intersection. The linearized Poincar\'e map gives us a dynamical description of the volume density on these components. To describe how, let's first recall the invariant volume density on the energy hypersurface $\tilde{H}^{-1}_Z(E)$.

\begin{defn}
	Let $\Omega$ denote the volume form on $\cN_{\cO}$ induced by the symplectic form and equip $\cN_{\cO}$ with the Riemannian metric $h_0$ induced from the one on $\cN\cong T^*_+P|_{P_0}\cong T^*P_0\setminus 0$ and the $\Ad$-invariant inner product on $\fg$. Using this metric we can define the gradient $\nabla H$ and the $2(n+\ell)-1$-form on $\cN_{\cO}$:
	\[ |\nabla\tilde{H}_Z|_{h_0}^{-2} \nabla\tilde{H}_Z \llcorner \Omega. \]
	Denote:
	\[ \nu_E:= \mbox{ the pullback of the above form to } \tilde{H}_Z^{-1}(E). \]
	The $\nu_E$ is invariant under the Hamiltonian flow $\tilde{\Phi}_t^Z$ and its absolute value $|\nu_E|$ defines an invariant measure on the energy hypersurface $\tilde{H}_Z^{-1}(E)$.
\end{defn}

\begin{lem}{\label{isolated_periods}}
	Under the clean intersection hypothesis, the fibered product $\fY_E$ comes equipped with a natural volume density. Consider then the case $\fY_E$ is a union of $\{0\}\times \tilde{H}_Z^{-1}(E)$ and finitely many disjoint isolated orbits:
	\[ Y_1:=\{(T_1,\tilde{\Phi}_t^Z(\gamma_1)) \ : \ t\in [0,T_1]\}, \ ..., \ Y_q:= \{(T_q,\tilde{\Phi}^Z_t(\gamma_q)) \ : \ t\in [0,T_q]\} \]
	with $T_j\neq 0$ for all $j$. Then the Poincar\'e first return map of each $\gamma_j$ is invertible and if $\Omega_{\gamma}$ is the symplectic volume form on $T^*_{\gamma}\cN_{\cO}$ then the induced volume density on $T^*_{\gamma}Y_j$ is given by:
	\begin{equation}\label{vol_density} |\det(I-P_{\gamma})|^{-1/2}|\nu_E| \end{equation}
	with $P_{\gamma}$ the Poincar\'e first return map for one, hence any, choice of symplectic local transversal $S$.
\end{lem}
\begin{proof}
	Indeed this follows immediately from the expression for the tangent space of $Y$ derived in the proof of \ref{set_of_orbits}, noticing that the constraint
	\[ \tau\frac{d}{dt}\Big|_{t=T}\tilde{\Phi}_t^Z(\gamma) + D\tilde{\Phi}^Z_T(\zeta) = \zeta \]
	in the case of isolated periodic orbits is such that $\zeta \mapsto \zeta - D\tilde{\Phi}^Z_T(\zeta)$ has a 1-dimensional kernel in $T_{\gamma}\tilde{H}^{-1}_Z(E) = (\nabla \tilde{H}_Z(\gamma))^{\perp}$ spanned the vector field $\tilde{Z}$ corresponding to the reduced flow $\tilde{\Phi}^Z_t$. Thus by \ref{floquet} the Poincar\'e first return map is invertible and the induced volume form on $TY_j$ is determined by the invariant volume form $\nu_E$ on $\tilde{H}_Z^{-1}(E)$ and the Poincar\'e first return map acting on
	\[ T\cN_{\cO}/\Span\{\tilde{Z},\nabla\tilde{H}_Z\} \cong T\tilde{H}_Z^{-1}(E)/\Span\{\tilde{Z}\} \]
	yielding the formula \ref{vol_density}.
\end{proof}

Next, let's discuss the phase associated to a periodic orbit. For this we need the following basic result from representation theory.

\begin{defn}
	The coadjoint orbit $\cO\subseteq\fg^*$ is called \textbf{integral} if and only if the cohomology class $[\omega_{\cO}]$ of its symplectic form $\omega_{\cO}$ is in the image of $H^2(\cO;\bZ)\to H^2(\cO;\bR)\cong H^2_{dR}(\cO;\bR)$.
\end{defn}

\begin{lem}
	A coadjoint orbit $\cO=G\cdot\xi_0$ is integral if and only if there exists a character
	\[ \chi_{\xi_0}:G_{\xi_0}\to \Un(1) \]
	such that
	\[ (d\chi_{\xi_0})_I=2\pi i\langle\xi_0,-\rangle:\fg_{\xi_0}\to i\bR \]
	where $I\in G$ is the identity matrix.
\end{lem}

So, when our coadjoint orbit is integral we have a $\Un(1)$-bundle defined by the character:
\[ G\times_{\chi_{\xi_0}}\Un(1)  \to \cO \]
where $G\times_{\chi_{\xi_0}}\Un(1)$ is the quotient of $G\times \Un(1)$ by the relation
\[ (g,z)\sim (gh,\chi_{\xi_0}(h^{-1})z) \ \mbox{ for all } \ h\in G_{\xi_0}. \]
The right $G$-action on $G$ yields a right $G$-action on the total space $G\times_{\chi_{\xi_0}}\Un(1)$ since the stabilizer $G_{\xi_0}$ is a normal subgroup. Through this, we identify every tangent space of the total space with the tangent space at the equivalence class $[I,1] \in G\times_{\xi_0}\Un(1)$ of $(I,1)\in G\times \Un(1)$.

\begin{lem}
	We have a natural isomorphism
	\[ T_{[I,1]}(G\times_{\chi_{\xi_0}}\Un(1))\cong \left\{ \left(Y+c\xi_0^{\#}, \ \frac{i}{2\pi}c\right) \ : \ Y\in \fg_{\xi_0}^{\perp}, \ c\in \bR \right\}. \]
	Furthermore, there is a principal $\Un(1)$-connection $\alpha$ on $G\times_{\xi_0}\Un(1)$ such that $d\alpha = \omega_{\cO}$ and, under our above isomorphism, it is given by:
	\[ \alpha\left(Y+c\xi_0^{\#}, \ \frac{i}{2\pi}c\right) = c=\langle \xi_0,Y+c\xi_0^{\#}\rangle. \]
\end{lem}

This $G$-equivariant bundle with connection over $\cO$ gives us a natural $\Un(1)$-bundle with connection over the reduced phase space $\cN_{\cO}$, which we describe now.

\begin{defn}{\bf The $\Un(1)$-Bundle With Connection: Construction I}{\label{connection}} \\
	Recalling that $\mu_{\cO}:\cN\times \bar{\cO}\to \fg^*$ we can consider the $G$-equivariant $\Un(1)$-bundle:
	\[ \left( \cN\times (G\times_{\chi_{\xi_0}}\Un(1))\right)\big|_{\mu_{\cO}^{-1}(0)} \to \mu_{\cO}^{-1}(0). \]
	If $\alpha^0$ denotes the Liouville 1-form on $\cN$ and $i:\mu_{\cO}^{-1}(0)\hookrightarrow \cN$ the inclusion then we have a $G$-invariant 1-form on the total space of this bundle given by:
	\[ i^*(\alpha^0 - \alpha). \]
	We then set:
	\[ Z_{\cO}:= \left( \cN\times (G\times_{\chi_{\xi_0}}\Un(1))\right)\big|_{\mu_{\cO}^{-1}(0)} \big/ G \to \mu_{\cO}^{-1}(0)/G =\cN_{\cO} \]
	with connection 1-form
	\[ \alpha_{\cO}:=\mbox{ the reduction of } i^*(\alpha^0-\alpha) \mbox{ mod } G. \]
\end{defn}

\begin{defn}{\bf The $\Un(1)$-Bundle with Connection: Construction II} \\
	Here we instead extend our right $G_{\xi_0}$-action on $\mu^{-1}(\xi_0)$ so $\mu^{-1}(\xi_0)\times \Un(1)$ via the character $\chi_{\xi_0}$. We then set
	\[ Z_{\cO}:= \left( \mu^{-1}(\xi_0)\times \Un(1)\right)\big/ G_{\xi_0} \to \mu^{-1}(\xi_0)/G_{\xi_0} = \cN_{\cO} \]
	with connection 1-form
	\[ \alpha_{\cO}:= \mbox{ the reduction of } i^*\alpha^0 +d\theta \mbox{ mod } G_{\xi_0} \]
	where now $i:\mu^{-1}(\xi_0) \hookrightarrow \cN$ is the inclusion.
\end{defn}

Finally we arrive at the holonomies that describe the quantum phase translation that occurs upon traveling along a classical periodic orbit.

\begin{defn}
	Let $\gamma:[0,T]\to \cN_{\cO}$ be a periodic orbit of the $\tilde{\Phi}^Z_s$-flow (i.e. $\gamma(s)=\tilde{\Phi}^Z_s(z_0)$ for some $z_0$ and $\gamma(0)=\gamma(T)$) and assume that $T$ is the minimum period of $\gamma$. We denote:
	\[ \Hol_{\cO}(\gamma):= \mbox{ the holonomy of } \alpha_{\cO} \mbox{ about the loop } \gamma. \]
	A key point is that while our construction of the $\Un(1)$-bundle with connection relied on a choice of character as well as a choice of $\xi_0\in\cO$, the element $\Hol_{\cO}(\gamma) \in \Un(1)$ is independent of these choices.
\end{defn}

The following proposition is from \cite{guillemin1990reduction} section 4. Their result applies here since it applies in the general context of symplectic reduction along an integral coadjoint orbit.

\begin{prop}
	The map $\Hol_{\cO}:\fY_E \to \Un(1)$ is locally constant. Furthermore, if we consider the symplectomorphism $\cN_{\cO}\cong \mu^{-1}(\xi_0)/G_{\xi_0}$ and suppose we had $\gamma \in \mu^{-1}(\xi_0)$ with $H_Z(\gamma)=E$ and $T\in\bR$, $g\in G_{\xi_0}$ such that $\Phi^Z_T(\gamma) = \gamma\cdot g$ then if $[\gamma] \in \cN_{\cO}$ denotes the image in the quotient we have:
	\[ \Hol_{\cO}(T,[\gamma]) = \chi_{\xi_0}(g)e^{iTE}.\]
\end{prop}

\section{The Wave Equation on a Kaluza-Klein Spacetime}{\label{section_3}}

Similar to the classical phase space, the quantum-mechanical phase space is the space of solutions to the equations of motion. Usually, for quantum particles in a classical gauge field, one solves Schr\"odinger's equations for sections of a vector bundle with connection. A choice of such a vector bundle corresponds to a choice of representation (usually irreducible) and hence a choice of fixed ``charge''. \\

When performing semiclassical asymptotics, one doesn't simply send $\hbar\to 0$ since $\hbar$ is a dimensional quantity, but instead sends an observable such as $\hat{S}/\hbar$ or $\hat{J}/\hbar$ to infinity (here $\hat{S}$ and $\hat{J}$ are respectively action and angular momentum). We will work as in \cite{hogreve1983classical},\cite{guillemin1990reduction} and send ``charge'' to infinity while holding the ratio of charge to energy fixed. Thus we need a quantum phase space which allows for varying representations of our structure group $G$. The relativistic version of this is defined below.

\begin{defn}
	Fix a smooth function $V\in C^{\infty}(M,\bR)$ satisfying $\cL_Z V =0$ to act as a time-independent potential. We denote
	\[ \Box_{\omega} := d^*d +V\circ \pi \]
	acting on $C^{\infty}(P,\bC)$. We also define operators:
	\[ D_Z := \frac{1}{i}\cL_{Z^{\omega}} \ \mbox{ and } \ D_{\xi}:= \frac{1}{i}\cL_{\hat{\xi}}. \]
\end{defn}

\begin{lem}{\label{invariance}}
	We have
	\[ [\Box_{\omega},D_Z]=0=[\Box_{\omega},D_{\xi}] \ \mbox{ for all } \xi\in \fg. \]
\end{lem}
\begin{proof}
	Indeed, isometries and hence Lie derivatives along Killing vector fields commute with the wave operator $d^*d$ so it suffices to shown that $D_Z(V\circ \pi) = 0 = D_{\xi}(V\circ \pi)$ for all $\xi\in\fg$. $D_Z(V\circ \pi)=0$ since $\cL_Z V=0$ and $Z^{\omega}$ is the horizontal lift of $Z$ to $P$. $D_{\xi}(V\circ \pi)=0$ since $V\circ \pi$ is constant on the fibers of $P$ and the vector fields $\hat{\xi}$ are vertical.
\end{proof}

It should be noted that these results are of interest even when $V=0$. Nevertheless, we include the potential term in order to allow our results to apply to the conformal wave equation
\[ d^*d\phi + C_nS_{g_{\omega}}\phi =0\]
for $C_n$ a dimensional constant and $S_{g_{\omega}}$ the scalar curvature of $g_{\omega}$. The origin of this variant of the wave equation comes from considering conformal variations $\tilde{g}_{\omega} := e^{2f}g_{\omega}$ of the Hilbert-Einstein action. Indeed, setting $\phi:= e^{(n-2)f/2}$ one can compute:
\[ \int_P S_{\tilde{g}_{\omega}} dV_{\tilde{g}_{\omega}} = \int_P \phi(d^*d\phi + C_n S_{g_{\omega}}\phi) dV_{g_{\omega}}. \]
It's also worth noting that if $S_G$ denotes the (constant) scalar curvature of the fibers of $P$ then from \cite{bleecker2005gauge} Theorem 9.3.7 we have
\[ S_{\tilde{g}_{\omega}} = S_g\circ \pi + \frac{1}{2}|F_{\omega}|^2\circ \pi + S_G\]
and so the scalar curvature of $(P,g_{\omega})$ does indeed satisfy our assumptions on the potential. \\

Returning to our operator $\Box_{\omega}$, the vanishing of our commutators with $D_Z$ and $D_{\xi}$ tells us that the operators $D_Z$ and $D_{\xi}$ leave the kernel
of $\Box_{\omega}$ invariant. In fact, we want to complete the kernel
of $\Box_{\omega}$ to a Hilbert space of sorts since this is the
quantum mechanical phase space. Indeed, the phase space in either
classical or quantum mechanics is most naturally viewed as the space
of solutions to the equations of motion (from a relativistic point of
view). Any choice of Cauchy hypersurface then provides a natural
identification of this phase space with a cotangent bundle; a more common non-relativistic
description of phase space. \\

As is done in \cite{strohmaier2018gutzwiller}, we adapt several definitions and results from \cite{bar2007wave}.

\begin{defn}
	For $T>0$ we denote:
	\[ \FE(P_{|t|\leq T}) := W^{1,2}([-T,T],L^2(P_0))\cap L^2([-T,T],
	W^{1,2}(P_0))\]
	and we will often interpret elements of $\FE(P_{|t|\leq T})$ as
	$\bC$-valued functions on $P_{|t|\leq T} := \pi^{-1}( [-T,T]\times
	\Sigma_0)$. For the moment, let's write:
	\[ \ker_T(\Box_{\omega}):= \{\phi\in \FE(P_{|t|\leq T}) \ | \
	\Box_{\omega}\phi=0\}.\]
\end{defn}

\begin{lem}{\cite{bar2007wave}}
	For $0<T_1<T_2$, the natural restriction map
	\[ \ker_{T_2}(\Box_{\omega}) \to \ker_{T_1}(\Box_{\omega}) \]
	is an isomorphism of locally convex spaces.
\end{lem}

\begin{defn}
	We denote
	\[ \ker\Box_{\omega} := \{\phi \in L^2_{loc}(P) \ : \ \phi|_{P_{|t|\leq T}} \in \ker_T(\Box_{\omega}) \mbox{ for all } T>0\}. \]
\end{defn}

\begin{lem}{\cite{bar2007wave}}
	We have $\phi\in\ker\Box_{\omega}$ if and only if $\phi\in L^2_{loc}(P)$ and there exists a $T>0$ such that $\phi|_{P_{|t|\leq T}}\in \ker_T(\Box_{\omega})$. Furthermore, the restriction map
	\[ \ker\Box_{\omega} \to \ker_T(\Box_{\omega}) \]
	is a vector space isomorphism for all $T>0$ and the locally convex topology on $\ker\Box_{\omega}$ obtained by declaring this to be a homeomorphism is independent of our choice of $T>0$.
\end{lem}

\begin{lem}{\cite{bar2007wave}}
	For each $t\in \bR$ the map
	\begin{align*}
		\CD_t:\ker\Box_{\omega} &\to W^{1,2}(P_t)\oplus L^2(P_t) \\
		\phi &\mapsto (\phi|_{P_t}, \ (\cL_{\hat{n}}\phi)|_{P_t})
	\end{align*}
	is an isomorphism of locally convex spaces (recall: $\nu$ is the future-directed unit normal of the Cauchy hypersurfaces $P_t$). Thus $\ker\Box_{\omega}$ has the topology of a Hilbert space.
\end{lem}

Our goal is to study the semiclassical asymptotics of the action of
time translation on the space $\ker\Box_{\omega}$.
\begin{defn}
	We denote by
	\[ e^{-itD_Z}:\ker\Box_{\omega}\to \ker\Box_{\omega}\]
	the isomorphism given by precomposing functions $\phi \in
	\ker\Box_{\omega}$ with the time $-t$ flow $P\to P$ along the Killing
	vector field $Z^{\omega}$.
\end{defn}
At the moment $e^{-itD_Z}$ is merely a
notation since we do not have a preferred Hilbert space inner product
with which to perform a functional calculus. Let's now describe how we perform the quantum mechanical analogue of symplectic reduction. \\

Towards this end, we should notice that $e^{-itD_Z}$ is not unitary on
$\ker\Box_{\omega}$ with respect to any of the Hilbert space
structures defined by a fixed Cauchy-data isomorphism since $Z^{\omega}\neq \hat{n}$. Instead, we proceed as in \cite{strohmaier2018gutzwiller} and 
notice that the equation $\Box_{\omega}\phi=0$ arises from a
variational problem and therefore has an associated stress-energy tensor.

\begin{defn}
	Given $\phi\in\ker\Box_{\omega}\cap C^{\infty}(P,\bC)$ we define the \textbf{stress-energy tensor} of $\phi$ to be
	\[ T(\phi):= d\phi\otimes d\phi -\frac{1}{2}\left( |d\phi|_{g_{\omega}}^2 + |\phi|^2V\right) \cdot g_{\omega}. \]
\end{defn}

The proof of the next few results are in \cite{strohmaier2018gutzwiller} but we sketch them here since the computations will be useful to us later.

\begin{lem}{\cite{strohmaier2018gutzwiller}}
	For $\phi\in \ker\Box_{\omega}\cap C^{\infty}(P,\bC)$, the stress-energy tensor $T(\phi)$ has divergence $-\frac{1}{2}\phi^2dV$ with respect to the metric $g_{\omega}$.
\end{lem}
\begin{proof}
	We compute, using that the metric is divergence free to get:
	\begin{align*}
		\Div\left( d\phi\otimes d\phi - \frac{1}{2}\left( |d\phi|_{g_{\omega}}^2 + \phi^2 V\right)g_{\omega}\right) &=-2(d^*d\phi)d\phi \\ & \ \ \ \ \ -\frac{1}{2}\left( 2\langle \nabla^2 \phi,\nabla\phi\rangle_{g_{\omega}} + 2\phi V\nabla \phi + \phi^2\nabla V\right) \llcorner g_{\omega} \\
		&= -(\Box_{\omega}\phi)d\phi -\frac{1}{2}\phi^2 dV
	\end{align*}
	and since $\Box_{\omega}\phi=0$ by assumption we're done.
\end{proof}

As such, we can use the stress-energy tensor to define a quadratic form on $\ker\Box_{\omega}\cap C^{\infty}(P,\bR)$ and extend it to a Hermitian form on $\ker\Box_{\omega}$ via the polarization identity.

\begin{defn}
	For $\phi \in \ker\Box_{\omega}\cap C^{\infty}(P,\bR)$ we denote
	\[ Q_{\omega}(\phi):= \int_{P_0} T(\phi)(Z^{\omega},\hat{n})dV_{P_0} \]
	and extend this quadratic form $Q_{\omega}$ to a Hermitian one via the polarization identity.
\end{defn}

\begin{lem}
	$Q_{\omega}$ is invariant under both the action of $G$ and $e^{-itD_Z}$.
\end{lem}
\begin{proof}
	First we recall the proof of $e^{-itD_Z}$-invariance from \cite{strohmaier2018gutzwiller}. Since $e^{-i(t_1+t_2)D_Z} = e^{-it_1D_Z}e^{-it_2D_Z}$ it suffices to show that the time derivative of $Q(e^{-itD_Z}\phi)$ vanishes at $t=0$. We do this for $\phi$ real-valued. Writing $\ast$ for the Hodge-$\ast$ on $P$ (not on $P_0$!) we can compute using Cartan's formula for the Lie derivative:
	\begin{align*}
		\frac{d}{dt}\Big|_{t=0} Q(e^{-itD_Z}\phi) &=\frac{d}{dt}\Big|_{t=0} \int_{P_0} \ast T(\phi)(Z^{\omega}) \\
		&= \int_{P_0} Z\llcorner d(\ast T(\phi)(Z)) \ \mbox{ since the pullback of an exact form is exact} \\
		&= \int_{P_0}N\Div_P(T(\phi)(Z)) dV_{P_0}
	\end{align*}
	where $\Div_P$ denotes the full divergence on $P$. From the previous lemma we have $\Div_P(T(\phi))(Z)=\frac{1}{2}\phi^2\cL_{Z^{\omega}}V =0$ and since Killing vector fields are divergence-free it follows that
	\[ \Div_P(T(\phi)(Z^{\omega})) = (\Div_P(T(\phi)))(Z^{\omega})=0 \]
	as desired. \\
	
	Next let's look at the $G$-action. We write $\phi\cdot g$ for the function $x\mapsto \phi(xg^{-1})$ and also continue to use the notation $\zeta\cdot g$ for the induced right action of $G$ on covectors $\zeta$. Since $G$ acts by isometries we have $|d(\phi\cdot g)|_{g_{\omega}}^2 = |d\phi|_{g_{\omega}}^2 \cdot g$ and therefore
	\[ T(\phi\cdot g)(Z^{\omega},\hat{n}) = T(\phi)(g^{-1}\cdot Z^{\omega},g^{-1}\cdot \hat{n})\cdot g. \]
	But $Z^{\omega} = \partial_t$ and $\hat{n}=N^{-1}(\partial_t - \beta)$ are both invariant under the $G$-action so
	\[ T(\phi\cdot g)(Z^{\omega},\hat{n}) = T(\phi)(Z^{\omega},\hat{n})\cdot g. \]
	Finally, the volume form $dV_{P_0}$ is invariant under the $G$-action since it is an action by isometries hence we can perform the change of variables $x\mapsto xg^{-1}$ in the integral to get
	\[ \int_{P_0} T(\phi\cdot g)(Z^{\omega},\hat{n})dV_{P_0} = \int_{P_0} T(\phi)(Z^{\omega},\hat{n})\cdot g dV_{P_0} = \int_{P_0}T(\phi)(Z^{\omega},\hat{n})dV_{P_0} \]
	as desired.
\end{proof}

Unfortunately: since we have allowed possibly negative potentials $V$ our quadratic form $Q_{\omega}$ need not be positive definite. Just as in \cite{strohmaier2018gutzwiller}, we apply several results on the general theory of Pontryagin and Krein spaces \cite{langer1982spectral},\cite{dritschel1996operators}. These are ``Hilbert spaces'' for which the inner product is permitted to have finite dimensional negative-definite and/or degenerate subspaces. As we will see below in \ref{positive_definite}, we only care about the operators $D_Z,D_{\xi}$, etc. on certain closed subspaces $\cH_m$ of $\ker\Box_{\omega}$ and $Q_{\omega}$ will be positive definite on these subspaces.

\begin{lem}{\cite{strohmaier2018gutzwiller}} \\
	$\ker Q_{\omega}\subseteq \ker\Box_{\omega}$ is finite
	dimensional and consists of $C^{\infty}$ functions.
	Furthermore, for all $\phi\in \ker Q_{\omega}$ we have:
	\[ D_Z\phi =0.\]
	In particular, if $\tilde{Q}_{\omega}$ is the Hermitian form
	on $\ker\Box_{\omega}/\ker Q_{\omega}$ induced by $Q_{\omega}$
	then $(\ker\Box_{\omega}/\ker Q_{\omega})$ is a Pontryagin
	space.
\end{lem}

\begin{lem}
	$D_Z$ descends to a
	Krein-self-adjoint operator on $\ker\Box_{\omega}/\ker
	Q_{\omega}$ whose domain contains the dense $G$-invariant
	subspace given by the image of $\ker\Box_{\omega}\cap
	C^{\infty}(P)$ in the quotient. Furthermore, the spectrum of
	$D_Z$ on this Krein space is discrete consisting of eigenvalues of finite multiplicity, invariant under
	$\lambda\mapsto \bar{\lambda}$ and $\lambda\mapsto -\lambda$,
	accumulates at $\pm\infty$ only, and has only finitely many
	non-real eigenvalues.
\end{lem}
\begin{proof}
	The only part of this not proven in \cite{strohmaier2018gutzwiller} was the $G$-invariance of the subspace $\ker\Box_{\omega}\cap C^{\infty}(P)$. However, the $G$-action preserves $C^{\infty}(P)$ since it is smooth and therefore preserves $\ker\Box_{\omega}\cap C^{\infty}(P)$ by \ref{invariance}.
\end{proof}

\begin{lem}{\cite{strohmaier2018gutzwiller},\cite{langer1982spectral},\cite{dritschel1996operators}}
	Let $\tilde{Q}_{\omega}$ denote the induced quadratic form on the quotient $\ker\Box_{\omega}/\ker Q_{\omega}$. Then there exists a maximal negative definite subspace
	\[ \tilde{V}^- \subseteq (\ker\Box_{\omega}/\ker Q_{\omega}, \ \tilde{Q}_{\omega}) \]
	which is invariant under $D_Z$ and $e^{-itD_Z}$. Furthermore, it is finite-dimensional with dimension an invariant of the Krein space and $D_Z$ themselves. Finally, $\tilde{V}^-$ is invariant under the action of $G$.
\end{lem}
\begin{proof}
	The only part of this not proven in the above-cited papers is the $G$-invariance. Indeed, suppose for contradiction that there was some $g\in G$ and $v\in \tilde{V}^-$ with $g\cdot v\notin \tilde{V}^-$. Consider the subspace $\tilde{W}^-:= g\cdot \tilde{V}^-$. Then, as $g\cdot v \notin \tilde{V}^-$ we have that the subspace $\tilde{W}^-+ \tilde{V}^-$ properly contains $\tilde{V}^-$. Furthermore, it is invariant under both $D_Z$ and $e^{-itD_Z}$ since $D_Z$ commutes with the $G$-action. Finally, $\tilde{Q}_{\omega}$ is negative-definite on $\tilde{W}^-$ since it is negative-definite on $\tilde{V}^-$ and invariant under the $G$-action, hence $\tilde{Q}_{\omega}$ is negative definite on $\tilde{W}^- + \tilde{V}^-$, contradicting maximality.
\end{proof}

Since $\tilde{Q}_{\omega}$ is non-degenerate and invariant under both the $G$-action and $e^{-itD_Z}$ we obtain the following immediate corollary.

\begin{cor}
	The subspace
	\[ \tilde{V}^+ := (\tilde{V}^-)^{\perp \tilde{Q}_{\omega}} \]
	is a Hilbert space with inner product $\tilde{Q}_{\omega}$, and is equipped with a unitary representation of $\bR\times G$ given by the restriction of $e^{-itD_Z}$ and the $G$-action from above.
\end{cor}

We can now begin the process of showing that $Q_{\omega}$ is positive definite on isotypic subspaces for irreducible representations with sufficiently large dominant integral weights.

\begin{lem}
	Let $V^-$ be the preimage of $\tilde{V}^-$ in $\ker\Box_{\omega}$ under the quotient map $\ker\Box_{\omega}\to \ker\Box_{\omega}/\ker Q_{\omega}$. Then $V^-$ is finite dimensional and contains $\ker Q_{\omega}$.
\end{lem}
\begin{proof}
	Indeed the quotient map restricts to a map $V^- \to \tilde{V}^-$ with kernel $\ker Q_{\omega}$. Choosing a splitting of this linear surjection gives us an isomorphism of vector spaces $V^-\cong \tilde{V}^-\oplus \ker Q_{\omega}$ and since $\tilde{V}^-\oplus \ker Q_{\omega}$ so is $V^-$.
\end{proof}

\begin{defn}
	For $\cO$ our integral coadjoint orbit and $m\in\bZ_{\geq 1}$ we let $\kappa_m$ denote the irreducible representation corresponding to the integral coadjoint orbit $m\cO\subseteq \fg^*$.
\end{defn}

\begin{prop}{\label{positive_definite}}
	There exists an $m_0 \in \bZ_{\geq 1}$ depending only on $\cO$, $D_Z$ and the Krein space $(\ker\Box_{\omega}/\ker Q_{\omega},\tilde{Q}_{\omega})$ such that for any $m\geq m_0$ and any $\phi\in \ker\Box_{\omega}$ which generates a cyclic $G$-representation $V_{\phi}\subseteq \ker\Box_{\omega}$ isomorphic to $\kappa_m$ we have
	\[ V_{\phi}\cap V^- = \{0\}. \]
	Thus for each $m\geq m_0$ we have a closed subspace
	\[ \cH_m := \bar{\Span_{\bC}\{ \phi\in\ker\Box_{\omega} \ : \ V_{\phi}\cong \kappa_m \}} \]
	on which $Q_{\omega}$ restricts to a positive definite Hilbert space inner product. Furthermore, our representation of $\bR\times G$ arising as the product of the $G$-action and $e^{-itD_Z}$ leaves $\cH_m$ invariant and is unitary.
\end{prop}
\begin{proof}
	Let $\phi \ker \Box_{\omega}$ generate a cyclic $G$-representation $V_{\phi}$ isomorphic to $\kappa_m$. Suppose that $V_{\phi}\cap V^- \neq \{0\}$ and so there existed a non-zero $\psi\in V_{\phi}\cap V^-$. Since $V^-$ is a $G$-invariant subspace we have $V_{\psi}\subseteq V^-$ where $V_{\psi}$ is the cyclic $G$-representation generated by $\psi$. Furthermore, $0\neq V_{\psi}\subseteq V_{\phi}$ and since $V_{\phi}$ is irreducible it follows that $V_{\psi}=V_{\phi}$. So it follows that:
	\[ \mbox{if } V_{\phi}\cong \kappa_m \ \mbox{ and } \ V_{\phi}\cap V^- \neq \{0\} \ \mbox{ then } \ V_{\phi}\subseteq V^-. \]
	Since $V^-$ is finite dimensional this can happen for at most finitely many irreducible cyclic invariant subspaces and hence for at most finitely many $m$. In fact, since the dimension of $V^-$ is an invariant of $D_Z$ and the Krein space $\ker\Box_{\omega}/\ker Q_{\omega}$ it follows that for $m_0$ large enough (with dependence as in the statement of the proposition) and all $m\geq m_0$ we have:
	\[ \mbox{if } \phi\in\ker\Box_{\omega} \mbox{ with } V_{\phi}\cong \kappa_m \mbox{ then } V_{\phi}\cap V^- = \{0\}. \]
	In particular, for $m\geq m_0$ and $\cH_m$ defined as in the statement of the proposition, $Q_{\omega}$ is positive definite on $\cH_m$. \\
	
	To show that our $\bR\times G$ action leaves $\cH_m$ invariant and is unitary it suffices to show that it leaves $\Span_{\bC}\{\phi\in\ker\Box_{\omega} \ : \ V_{\phi}\cong \kappa_m\}$ invariant and is unitary here, since it will then extend to $\cH_m$ by uniform continuity. Since $Q_{\omega}$ is invariant under the full $\bR\times G$-action, unitarity is immediate. All that remains is to check invariance. However, since $\kappa_m$ is irreducible it follows that for any $\phi$ with $V_{\phi}\cong \kappa_m$ and any $g\in G$ we have $0\neq V_{\phi\cdot g}\subseteq V_{\phi}$ hence $V_{\phi\cdot g}=V_{\phi}$ thus we have invariance, as desired.
\end{proof}

It is worth noting that, as remarked in \cite{strohmaier2018gutzwiller}, if $V\geq 0$ and there exists some $x\in \Sigma_0$ for which $V(x)>0$ then $Q_{\omega}$ is positive definite. This is especially true for the massive Klein-Gordon equation where $V$ is a positive constant. In \cite{strohmaier2020semi} the special case of our results where $G=\Un(1)$ and $(P,\omega)$ were trivial was considered. In this case it was shown that when projected down to $M$ our parameter $m\in\bZ_{\geq 1}$ above actually corresponds to mass. We will demonstrate an analogue of this later in \ref{vector_bundles}. \\

Another important remark is that not every $\phi\in \cH_m$ has $V_{\phi}\cong \kappa_m$. This is most easily seen in the Euclidean-signature case where $M$ is a single point. Then $P=G$ and our Hilbert space is $L^2(G)$ which, by the Peter-Weyl theorem, contains every irreducible representation of $G$ as a cyclic subspace. However, as was shown in \cite{greenleaf1971cyclic}, since $G$ is compact Hausdorff and second-countable, the entire representation $L^2(G)$ is itself a cyclic representation. \\

Combining our previous facts, for $m\geq m_0$ we can decompose:
\[ \cH_m = \bigoplus_{\ell\in\bZ}^{L^2} \cH_{m,\ell}\]
with $\cH_{m,\ell}$ the $\lambda_{m,\ell}$-eigenspace for $D_Z$ on $\cH_m$,
organized so that $\lambda_{m,\ell}\leq \lambda_{m,\ell+1}$ for all
$\ell\in\bZ$. If $\lambda_{m,\ell}=\lambda_{m,\ell+1}$ then
$\cH_{m,\ell}=\cH_{m,\ell+1}$ and otherwise these spaces are
orthogonal (this is the sense in which the above is indeed an
$L^2$-direct sum). We can then further decompose:
\[ \cH_{m,\ell} = \bigoplus_{j=1}^{\mu(m,\ell)} \kappa_m\]
and it is worth noticing that $\mu(m,\ell)$ is indeed always finite
since $\cH_{m,\ell}$ itself is finite dimensional (being an eigenspace
for $D_Z$). \\

Since we will be studying asymptotics as $m\to \infty$, there's no
harm in replacing $\kappa$ with $\kappa_{m_0}$ so that we may assume
$m_0=1$. As such, we want to study the time evolution of quantum
states in the subspace
\[ \cH:= \bigoplus^{L^2}_{m\geq 1} \cH_m \subseteq
\ker\Box_{\omega}.\]
However, we still haven't fully specified a direction in which to take
our large quantum numbers limit. Indeed, for fixed $m$ the eigenvalues
$\lambda_{m,\ell}$ very well might accumulate at $\pm\infty$ as $\ell$
tends to $\pm\infty$. Thus for each $E\in \bR$ we could consider
eigenvalues satisfying
\[ \lambda_{m,\ell}\sim mE\]
and different choices of $E$ might very well yield different
$m\to\infty$ asymptotics. Classically this is reflected in the fact
that symplectic reduction along $\cO$ generally leads to phase spaces
which are not conical. As such, our problem is broken into two steps:
\begin{enumerate}
	\item For $m$ fixed, ``count'' eigenvalues satisfying
	$\lambda_{m,\ell}\sim mE$.
	\item Understand the asymptotics of the above count as $m\to
	\infty$.
\end{enumerate}
The first step is fairly straight-forward. It is highly unlikely for
us to have any eigenvalues satisfying $\lambda_{m,\ell}=mE$ exactly and so we
instead sum over all $\ell\in\bZ$, weighting eigenvalues near $mE$ the
most. By stationary phase, this is described for large frequencies by
the distribution:
\[ \varphi\mapsto \Tr\left( \int_{-\infty}^{\infty} \varphi(t)
e^{-it(D_Z-mE)}|_{\cH_m} dt\right) = \sum_{\ell\in\bZ}
\hat{\varphi}(\lambda_{m,\ell}-mE) =:\mu(E,m,\varphi).\]
We use the letter $\mu$ to denote this distribution since it can be
viewed as a multiplicity for the representation on $\cH$ of $\bR\times
G$ associated to the coadjoint orbit $\{E\}\times \cO\subseteq
\bR\oplus \fg^*$. The point is that (modulo factors of $2\pi$),
$\hat{\varphi}$ approaches $\delta_0$ as $\varphi\to 1$ and so in
this limit the right hand side approaches the literal
multiplicity of $mE$ as an eigenvalue on $\cH_m$. However, this is
only a moral since the above limit does not converge.
Instead we first notice that $\mu(E,m,-)$ defines a linear functional on the collection of all $\varphi \in \cS(\bR)$ with compactly supported Fourier transform. Our goal now is to apply a result of \cite{islam2021gutzwiller} which generalizes the Weyl law of \cite{strohmaier2018gutzwiller} to vector bundles in order to prove that $\mu(E,m,-)$ is actually tempered and hence $\mu(E,m,\varphi)$ is defined for any $\varphi\in\cS(\bR)$. \\

\subsection{Relation to Vector Bundles}{\label{vector_bundles}}
We begin by recalling the well-known fact that for any unitary representation $V$ of $G$ there is an isomorphism
\[ C^{\infty}(P,V)^G\cong \Gamma(M,P\times_G V) \]
between $V$-valued $G$-equivariant smooth functions on $P$ and smooth sections of the associated vector bundle $P\times_G V$ over $M$. Furthermore, the Hermitian inner product on $V$ defines a Hermitian fiber metric on $P\times_G V$. We will need a less well-known, but related construction.

\begin{defn}
	We fix an $m\geq m_0$ so that $Q_{\omega}$ is positive definite on $\cH_m\subseteq \ker\Box_{\omega}$ and denote by $\kappa_m:G\to \Un(V_m)$ our irreducible representation corresponding to $m\cO$. We also let $d_m:= \dim_{\bC} V_m$ and fix an orthonormal basis $\vec{e}_1,...,\vec{e}_{d_m}$ for $V_m$, writing $\langle -,-\rangle_m$ for our Hermitian inner product on $V_m$.
\end{defn}

\begin{lem}{\label{Direction_1}}
	Let $\vec{\psi}\in C^{\infty}(P,V_m)^G$ and $\vec{v}\in V_m$ both be non-zero. Define a function
	\begin{align*}
		\phi:P &\to \bC \\
		\phi(p) &:= \langle \vec{\psi}(p),\vec{v}\rangle_m.
	\end{align*}
	Then $V_{\phi}\cong V_m$ as $G$-representations. Furthermore, if $\Box_{\omega}$ is extended to act on $V_m$-valued smooth functions it follows that $\Box_{\omega}\vec{\psi}=0$ if and only if $\Box_{\omega}\phi =0$.
\end{lem}
\begin{proof}
	Since $\vec{v}\in V_m$ is non-zero and $V_m$ is irreducible, it is a cyclic vector and so for each $j=1,...,d_m$ there are finitely many group elements $g_j^i \in G$ such that $\sum_i g_j^i \vec{v} = \vec{e}_j$. Thus
	\[ \sum_i \phi(p(g_j^i)^{-1}) = \sum_i \langle \vec{\psi}(p),g_j^i\vec{v}\rangle_m = \langle \psi(p),\vec{e}_j\rangle_m. \]
	So the functions $\langle \vec{\psi}(-),\vec{e_j}\rangle_m$ are in $V_{\phi}$ for all $j=1,...,d_m$. Furthermore every function $p\mapsto \psi(pg^{-1}) = \langle \vec{\psi}(p),g\vec{v}\rangle_m$ is in the span of the functions $\langle \vec{\psi}(-),\vec{e}_j\rangle_m$ hence
	\[ V_{\phi} = \Span_{\bC}\left\{ \langle\vec{\psi}(-),\vec{e}_1\rangle_m,...,\langle \vec{\psi}(-),\vec{e}_{d_m}\rangle_m \right\}. \]
	The set of functions $\langle \vec{\psi}(-),\vec{e}_j\rangle_m$ are linearly independent since if $a^j\in\bC$ are such that $\langle \vec{\psi}(p),a^j\vec{e}_j\rangle_m=0$ for all $p\in P$ then since $\vec{\psi}\neq 0$ there exists a $p\in P$ with $0\neq \vec{\psi}(p) \in V_m$. Since $V_m$ is irreducible there exists elements $g_k \in G$ such that $\sum_k g\vec{\psi}(p) = a^j\vec{e}_j$ and so
	\[ 0 = \sum_k \langle \vec{\psi}(pg_k^{-1}),a^j\vec{e}_k\rangle = \sum_j |a^j|^2 \]
	hence $a^j=0$ for all $j$ as desired. Therefore the map
	\[ \vec{e}_j \leftrightarrow \langle \vec{\psi}(-),\vec{e}_j\rangle_m \]
	induces an isomorphism of $G$-representations $V_m\cong V_{\phi}$. \\
	
	If $\Box_{\omega}\vec{\psi}=0$ then by definition ($\vec{v}$ and $\langle -,-\rangle_m$ are constant) $\Box_{\omega}\phi =0$. Conversely, $G$-invariance of $\Box_{\omega}$ implies that if $\Box_{\omega}\phi =0$ then $\Box_{\omega}f=0$ for all $f\in V_{\phi}$ and hence $\Box_{\omega}\langle \vec{\psi}(-),\vec{e}_j\rangle_m =0$ for all $j$. Therefore $\Box_{\omega}\vec{\psi}=0$ as desired.
\end{proof}

Usually one doesn't look at the full wave operator $\Box_{\omega}$ applied to $\vec{\psi}\in C^{\infty}(P,V)^G$ but only at the ``horizontal'' wave operator. To relate these two wave operators, we fix a root system for $\fg$ compatible with our $\Ad$-invariant inner product and let:
\[ \rho := \ \mbox{ the sum of all positive roots} \]
and
\[ \Lambda_0 := \ \mbox{ the dominant integral weight for } \kappa_{m_0}. \]

\begin{lem}
	The wave operator $\Box_{\omega}$ on $C^{\infty}(P)$ splits as a sum of vertical and horizontal parts:
	\[ \Box_{\omega} = \Box_H - \Delta_G \]
	where $\Box_H$ is the horizontal wave operator (plus the potential) and $\Delta_G$ is the Laplacian on the fibers. These operators commute and if $\phi \in \cH_m$ has $V_{\phi}\cong V_m$ then $\Delta_G$ acts on $V_m$ as multiplication by a constant. Hence $\Delta_G$ acts by multiplication by a constant on all of $\cH_m$ and this constant is given by:
	\[ \Delta_G|_{\cH_m} = \langle m\Lambda_0, \ m\Lambda_0 +\rho\rangle. \]
\end{lem}
\begin{proof}
	The existence of the splitting and the fact that $[\Box_H,\Delta_G]=0$ follows from \cite{guillemin1990reduction} section 6. Since $\Delta_{\omega}$ and $\Delta_H$ both commute with the $G$-action it follows that $\Delta_G$ does as well hence $\Delta_G$ does indeed preserve $V_{\phi}$. In fact, by the explicit form of $\Delta_G$ we see that its action on $V_{\phi}$ is precisely the action of the quadratic Casimir and hence is given by multiplication by $\langle m\Lambda_0, \ m\Lambda_0+\rho\rangle$.
\end{proof}

In fact, we see that $\Delta_G$ preserves our space
\[ \cH = \bigoplus_{m\geq m_0}^{L^2} \cH_m \]
and on this space $\cH_m$ is precisely the $\langle m\Lambda_0, \ m\Lambda_0 + \rho\rangle$-eigenspace of $\Delta_G$.

\begin{defn}
	We denote by $\Box_m$ the operator
	\[ \Box_m := \Box_H - \langle m\Lambda_0, \ m\Lambda_0 + \rho\rangle. \]
\end{defn}

\begin{lem}{\label{Direction_2}}
	Denote by
	\[ \Gr_m(P):= \{ V\subseteq \ker\Box_{\omega}\cap C^{\infty}(P) \ : \ V \mbox{ is } G\mbox{-invariant and } V\cong V_m \} \]
	the collection of all invariant subspaces of $\ker\Box_{\omega}$ which are isomorphic to $V_m$ as $G$-representations. Then for each $V\in \Gr_m(P)$ we have $V\subseteq \cH_m$. Furthermore if $\Phi:V_m\to V$ is any isomorphism of $G$-representations then
	\[ \vec{\psi}(p) := \sum_{j=1}^{d_m} \Phi(\vec{e}_j)(p)\vec{e}_j \]
	is a $G$-equivariant $V_m$-valued function with
	\begin{equation}\label{VB_KG} \Box_m\vec{\psi} = \Box_H\vec{\psi} - \langle m\Lambda_0, \ m\Lambda_0 + \rho\rangle \vec{\psi} =0. \end{equation}
	Finally, the definition of $\vec{\psi}$ is independent of our choice of orthonormal basis $\vec{e}_j$.
\end{lem}
\begin{proof}
	Since each $\Phi(\vec{e}_j)$ generates a cyclic representation isomorphic to $V_m$ it automatically follows that $V\subseteq \cH_m$ and $\vec{\psi}$ satisfies \ref{VB_KG}. So all that remains to be checked is $\vec{\psi}$'s equivariance and basis-independence. However since $\Phi$ is an isomorphism of $G$-representations we can compute:
	\[ \vec{\psi}(pg^{-1}) = \sum_{j=1}^{d_m} \Phi(\vec{e}_j)(pg^{-1})\vec{e}_j = \sum_{j=1}^{d_m} \Phi(g\vec{e}_j)(p)\vec{e}_j. \]
	But if we write $g\vec{e}_j = g^i_j\vec{e}_i$ then we arrive at:
	\[ \vec{\psi}(pg^{-1}) = \sum_{j=1}^{d_m} g^i_j \Phi(\vec{e}_i)(p)\vec{e}_j = \sum_{i=1}^{d_m}\Phi(\vec{e}_i) \ g\vec{e}_i \]
	proving equivariance. Similarly, if $\vec{f}_j \in V_m$ is another orthonormal basis then there exists a unitary matrix $A$ satisfying $\vec{e}_j = A_j^i\vec{f}_i$ hence
	\[ \vec{\psi}(p) = \sum_{j=1}^{d_m} A_j^iA_j^k \Phi(\vec{e}_i)(p)\vec{e}_k = \sum_{i=1}^{d_m} \Phi(\vec{e}_i)(p)\vec{e}_i \]
	as desired.
\end{proof}

Since $V_m$ is irreducible, Schur's lemma tells us that any two isomorphisms $V_m\cong V$ of $G$-representations differ by a multiplicative non-zero constant complex number. As such, we obtain the following corollary.

\begin{cor}
	There is a natural isomorphism
	\begin{align*}
		\Gr_m(P) &\to \{ \vec{\psi} \in C^{\infty}(P,V_m)^G \ : \ \Box_m\vec{\psi} =0\}/\bC^{\times} \\
		V &\mapsto \sum_{j=1}^{d_m} \Phi(\vec{e}_j)(-)\vec{e}_j \ \mbox{mod } \bC^{\times}
	\end{align*}
	where in the above expression $\vec{e}_j$ is any choice of orthonormal basis for $V_m$ and $\Phi$ is any choice of isomorphism of $G$-representations $V_m\cong V$.
\end{cor}
\begin{proof}
	This is simply a combination of \ref{Direction_1} and \ref{Direction_2}, taking care to remark that the two constructions from these two lemmas are inverse to one-another (taking $\vec{v}=\vec{e}_1$ in \ref{Direction_1}).
\end{proof}

Our final step is to compare elements of $C^{\infty}(P,V_m)^G$ with sections of the associated vector bundle.

\begin{defn}
	We define a map $\Psi:C^{\infty}(P,V_m)^G \to \Gamma(M,P\times_G V_m)$ as follows. Given $\vec{\psi}\in C^{\infty}(P,V_m)^G$ and $x\in M$ we choose an arbitrary $p\in P$ in the fiber over $x$ and define
	\[ \Psi(\vec{\psi})(x):= \mbox{ the equivalence class of } (p,\vec{\psi}(p)) \mbox{ in the fiber } (P\times_G V_m)_x. \]
	We recall from \cite{bleecker2005gauge} Chapter 3, for example, that $\Psi$ is an isomorphism. Furthermore there is an induced covariant derivative $\nabla^m$ on $P\times_G V_m$ which corresponds under $\Psi$ to the horizontal exterior derivative on $P$ with respect to $\omega$, and there is a Hermitian fiber metric $\langle -,-\rangle_m$ on $P\times_G V$ corresponding to the constant Hermitian inner product $\langle -,-\rangle_m$ on $V_m$.
\end{defn}

Now, let's let $V\in \Gr_m(P)$ and choose an isomorphism $\Phi:V_m \to V$ which is unitary where $V$ is given the $Q_m$-inner product. Writing
\[ \vec{\psi}(p):= \sum_{j=1}^{d_m} \Phi(\vec{e}_j)(p)\vec{e}_j \]
it follows that the expression
\[ Q_{\omega}(\vec{\psi}):= \sum_{j=1}^{d_m} Q_{\omega}(\Phi(\vec{e}_j)) \]
is independent of our choice of orthonormal basis $\vec{e}_j$ or unitary isomorphism $\Phi$. We also have the following explicit formula from \cite{strohmaier2018gutzwiller} where we use Greek $\mu,\nu,...$ for indices of coordinates tangent to $\Sigma_0\subseteq M$ and Roman $a,b,...$ indices for coordinates tangent to the fibers of $P_0$:
\begin{align*}
Q(\Phi(\vec{e}_j)) &= \int_{P_0} N^{-1} \Big( |\partial_t \Phi(\vec{e}_j)|^2 + (N^2h^{\mu\nu} - \beta^{\mu}\beta^{\nu})(\partial_{\mu}\Phi(\vec{e}_j))(\partial_{\nu} \bar{\Phi(\vec{e}_j)})  \\ & \ \ \ \ \ \ \ \ \ \ \ \ \ \ \ \ \ \ \ \ \ \ \ \ \ \ \ \ \ \ \ \ \ + \Tr\left( \omega(d\Phi(\vec{e}_j))\omega(d\bar{\Phi(\vec{e}_j)})^T\right) +|\Phi(\vec{e}_j)|^2 V\Big) dV_{P_0}
\end{align*}
By equivariance it follows that if $\xi_1,...,\xi_d$ is an orthonormal basis for $\fg$ then
\[ \omega(d\Phi(\vec{e}_j)) = \sum_a (\cL_{\hat{\xi}_a}\Phi(\vec{e}_j)) \hat{\xi}_a = \sum_a \Phi(\xi_a\cdot \vec{e}_j)\hat{\xi}_a \]
and so
\[ \Tr\left( \omega(d\Phi(\vec{e}_j))\omega(d\bar{\Phi(\vec{e}_j)})^T\right) = \sum_a |\Phi(\xi_a\cdot \vec{e}_j)|^2= \langle m\Lambda_0, \ m\Lambda_0+\rho\rangle |\Phi(\vec{e}_j)|^2. \]
Thus we obtain
\begin{align*}
Q(\Phi(\vec{e}_j)) &= \int_{P_0} N^{-1} \Big( |\partial_t \Phi(\vec{e}_j)|^2 + (N^2h^{\mu\nu} - \beta^{\mu}\beta^{\nu})(\partial_{\mu}\Phi(\vec{e}_j))(\partial_{\nu} \bar{\Phi(\vec{e}_j)}) \\ & \ \ \ \ \ \ \ \ \ \ \ \ \ \ \ \ \ \ \ \ \ \ \ \ \ \ \ \ \ \ \ \ \ \ \ \ \ \ \ + |\Phi(\vec{e}_j)|^2 \left( V+\langle m\Lambda_0, \ m\Lambda_0+\rho\rangle\right) \Big) dV_{P_0}
\end{align*}
Furthermore, from this explicit expression we see that the sum
\begin{align}\label{invariant} \sum_{j=1}^{d_m} N^{-1} \Big( |\partial_t \Phi(\vec{e}_j)|^2 &+ (N^2 h^{\mu\nu}-\beta^{\mu}\beta^{\nu})(\partial_{\mu}\Phi(\vec{e}_j))(\partial_{\nu}\bar{\Phi(\vec{e}_j)}) \\ & \ \ \ \ \ \ \ \ \ \ \ \ \ \ \ + |\Phi(\vec{e}_j)|^2 (V+\langle m\Lambda_0, \ m\Lambda_0 + \rho\rangle) \Big)\notag \end{align}
is invariant under the $G$ action.

\begin{defn}
	Given a section $s\in \Gamma(M,P\times_G V_m)$ we define the \textbf{bundle stress-energy tensor} $T_m(s)$ to be the symmetric 2-tensor on $M$ given by:
	\[ T_m(s)_{ij} := \langle \nabla_i^ms, \nabla_j^ms\rangle - \frac{1}{2}\left( |\nabla^ms|^2 + |s|^2\left(V+\langle m\Lambda_0, \ m\Lambda_0 + \rho\rangle\right)\right)g_{ij} \]
	where we recall that $\nabla^m$ is the covariant derivative on $P\times_G V_m$ induced by the connection $\omega$.
\end{defn}

Since $\langle m\Lambda_0, \ m\Lambda_0 + \rho\rangle$ is a constant and the connection $\nabla^m$ is compatible with the fiber metric it follows exactly as in the scalar case that if we abuse notation and also use $\Box_m$ to denote
\[ \Box_m = (\nabla^m)^*\nabla^m +V+\langle m\Lambda_0, \ m\Lambda_0 + \rho\rangle \]
acting on sections of $P\times_G V_m$ then
\begin{align*}
	\Div_M(T_m(s)) &= -\langle \Box_m s, \nabla^m s\rangle - \frac{1}{2} |s|^2 dV \\
	\Div_M(T_m(s)(Z)) &= (\Div_MT_m(s))(Z) =0 \ \mbox{ if } \Box_m s=0
\end{align*}
where we note that despite the raised and lowered $m$'s appearing, we are not summing over them: they merely denote the representation of $G$ we are considering. \\

Just as in the scalar case, we can define the space of finite-energy solutions $s$ to $\Box_m s=0$ and one has Cauchy-data isomorphisms:
\[ s \mapsto \left( s|_{\Sigma_0}, \ (\nabla^m_{\hat{n}}s)|_{\Sigma_0}\right) \]
which give $\ker\Box_m$ the topology of a Hilbert space. Furthermore, since $Z$ is Killing the covariant derivative $\nabla^m_Z$ commutes with $\Box_m$ and we have the densely defined operator
\[ D_{m,Z}:= \frac{1}{i}\nabla^m_Z:\ker\Box_m\cap C^{\infty}\to \ker\Box_m\cap C^{\infty}. \]
Combining all of our results in this section and especially using \ref{invariant} we arrive at the following result.

\begin{prop}
	Let $V\in \Gr_m(P)$ and $\Phi:V_m\to V$ a unitary isomorphism so that we can define
	\[ \vec{\psi}(p) := \sum_{j=1}^{d_m} \Phi(\vec{e}_j)(p)\vec{e}_j. \]
	Then $\Psi(\vec{\psi}) \in \ker\Box_m$ and
	\[ Q_m(\Psi(\vec{\psi})) := \int_{\Sigma_0} T_m(\Psi(\vec{\psi}))(Z,\hat{n})dV_{\Sigma_0} = \Vol(G)Q_{\omega}(\vec{\psi})\]
	where $\Vol(G)$ is taken with respect to the volume form induced by our $\Ad$-invariant inner product on $\fg$. Furthermore, since $m\geq m_0$ by assumption it follows that $Q_m$ is positive definite on the finite energy space $\ker\Box_m$.
\end{prop}

We are now ready to apply the results of \cite{islam2021gutzwiller}. Really we are using a very special case of these results since we only need them to show that our multiplicity distributions $\mu(E,m,-)$ are tempered.

\begin{thm}{\cite{islam2021gutzwiller}}
	The operator $D_{m,Z}$ is self-adjoint on $(\ker\Box_m,Q_m)$ with $\sigma(D_{m,Z})\subseteq \bR$ discrete and accumulating at $\pm\infty$ with polynomial growth.
\end{thm}

\begin{cor}
	The spectrum of $D_Z$ on $\cH_m$ is real, discrete and accumulates at $\pm\infty$ with polynomial growth. Furthermore, the multiplicity of $\lambda\in \sigma(D_Z)$ is equal to $d_m=\dim(V_m)$ times the multiplicity of $\lambda\in \sigma(D_{m,Z})$.
\end{cor}

\begin{cor}
	The distribution $\mu(E,m,-)$ given by
	\[ \mu(E,m,\varphi):= \sum_{\ell\in\bZ} \hat{\varphi}(\lambda_{m,\ell}-mE) \]
	is a tempered distribution on $\bR$. Here we recall that $\cdots \leq\lambda_{m,\ell}\leq \lambda_{m,\ell+1}\leq\cdots$ are the eigenvalues of $D_Z$ on $\cH_m$.
\end{cor}

\section{Proofs of Main Theorems}{\label{section_4}}
We are now prepared to study the $m\to\infty$ asymptotics of $\mu(E,m,\varphi)$. As it turns out, this will depend significantly on whether or not $0\in\supp\hat{\varphi}$. For now we illustrate the method from Section 7 of \cite{guillemin1989circular} where $\varphi$ is fixed and arbitrary. This method takes advantage of the periodicity and ``positive frequency'' property of our distributions to express them in terms of linear combinations of the basic homogeneous periodic distributions
\[ \sum_{m=1}^{\infty}m^k z^{-m}e^{im\theta} \]
with $z\in S^1$ and $k\in\bZ_{\geq 0}$ determining the location of the singularity and the homogeneity respectively. A key advantage of these techniques from \cite{guillemin1989circular} is that it circumvents the need for general Tauberian theorems. \\

From now on we replace $\cO$ with $m_0\cO$ so that we may assume $m_0=1$.

\begin{defn}
	We define the \textbf{generating function} of the multiplicities $\mu(E,m,\varphi)$ to be the periodic distribution in the real variable $\theta$:
	\[ \Upsilon(\varphi)(\theta):= \sum_{m=1}^{\infty} \mu(E,m,\varphi)e^{im\theta} \]
	defined for any function $f(\theta)$ which is the Fourier transform of a compactly supported function on $\bR$.
\end{defn}

Distributions of the form $\sum_{m=1}^{\infty} a_m e^{im\theta}$ with $a_m$ real are called \textbf{Hardy distributions}. These are precisely the distributions on the sphere $S^1$ whose negative Fourier coefficients all vanish and so they have nice descriptions in terms of boundary values of holomorphic functions on the unit disk via the Paley-Weiner theorem. The asymptotics of the Fourier coefficients of such distributions, especially when $a_m$ is a homogeneous function of $m$, have been studied in books such as \cite{monvel1981spectral} Sections 12 and 13, and applied to spectral asymptotics in \cite{guillemin1989circular},\cite{guillemin1990reduction},\cite{strohmaier2020semi} for example. \\

Later in this section we will write $\Upsilon(\varphi)$ as a composition of Fourier integral operators and through this we will show that it is actually in $\cD'(\bR/2\pi\bZ)$. For now we illustrate how the asymptotics of the Fourier coefficients of a general Lagrangian distribution $\Upsilon$ on $S^1$ can be related to its principal symbol.

\begin{defn}
	Let $\Upsilon \in \cD'(\bR/2\pi\bZ)$. An element $s_0\in \mbox{singsupp}(\Upsilon)$ is called \textbf{classical of degree $k$} if and only if when interpreting $\Upsilon$ as a $2\pi\bZ$-periodic distribution on $\bR$ we have:
	\begin{enumerate}
		\item $s_0$ is an isolated singularity, and
		\item for any $\rho\in C^{\infty}_c(\bR)$ with $\rho\equiv 1$ on a neighborhood of $s_0$ and $\singsupp(\Upsilon)\cap \supp(\rho) = \{s_0\}$ we have asymptotic expansions:
		\begin{align*}
			\hat{\rho\Upsilon}(\xi) &\sim e^{-is_0\xi}\sum_{\ell=0}^{\infty} c_{\ell}^+ \xi^{k-\ell} \ \mbox{ as } \ \xi\to +\infty \ \mbox{ and} \\
			\hat{\rho\Upsilon}(\xi) &\sim e^{-is_0\xi}\sum_{\ell=0}^{\infty} c_{\ell}^- \xi^{k-\ell} \ \mbox{ as } \ \xi \to -\infty.
		\end{align*}
	\end{enumerate}
\end{defn}

\begin{lem}
	Let $s_0\in\bR$ be a classical singularity of $\Upsilon$ of degree $k$, and let $\rho\in C^{\infty}_c(\bR)$ have $\rho\equiv 1$ on a neighborhood of $s_0$ and $\singsupp(\Upsilon(\varphi))\cap \supp(\rho) = \{s_0\}$. Then
	\[ \rho\Upsilon(\varphi) \in I^{k+1/4}(\bR,\Lambda) \]
	where $\Lambda = \{ (s_0,\xi) \in T^*\bR\setminus 0 \ : \ \xi\neq 0 \}$. If $c_{\ell}^- =0$ for all $\ell$ (in which case the singularity is called \textbf{positive}) then instead $\Lambda = \{(s_0,\xi) \in T^*\bR\setminus 0 \ : \ \xi >0\}$.
\end{lem}
\begin{proof}
	We can write the distribution $\rho\Upsilon$ as
	\[ \langle \rho\Upsilon,\psi\rangle = \int_{\bR} e^{i(s-s_0)\xi} (e^{is_0\xi}\hat{\rho\Upsilon}(\xi))\psi(s)dsd\xi \]
	and so it suffices to check whether the function $e^{is_0\xi}\hat{\rho\Upsilon}(\xi)$ lives in the correct symbol class. Since $\rho\Upsilon\in \cE'(\bR)$ its Fourier transform is a smooth function and our asymptotics precisely tell us that it lives in the symbol class
	\[ S^k(\bR_s\times \bR_{\xi}) \ \mbox{ (it is independent of } s). \]
	Since $\dim(\bR_s)=1=\dim(\bR_{\xi})$ this is the correct order for a symbol to define an FIO of order
	\[ k-(1-2\cdot 1)/4 = k+1/4. \]
	As for the Lagrangian, one simply notices first that the $\xi$-critical points of the phase are precisely the set of $(s,\xi)$ with $s=s_0$, meanwhile the support of $e^{is_0\xi}\hat{\rho\Upsilon}(\xi)$ is everywhere in the non-positive singularity case and is a positive ray in the case of a positive singularity.
\end{proof}

\begin{lem}
	Suppose $\Upsilon$ had only finitely many singularities $z_1,...,z_q\in S^1$ and that for $s_1,...,s_q\in [0,2\pi]$ with $e^{-is_1}=z_1,...,e^{-is_q}=z_q$ the singularities $s_1,...,s_q$ were all classical with respective degrees $k_1,...,k_q$. For some $\rho_j\in C^{\infty}_c(\bR)$ smooth cutoffs with $\rho_j\equiv 1$ on a neighborhood of $s_j$ and $\singsupp(\Upsilon)\cap \supp(\rho_j) = \{s_j\}$, and for $c^{\pm,j}_{\ell}$ the coefficients of our asymptotic expansions for $\hat{\rho_j\Upsilon}$:
	\[ \hat{\rho_j\Upsilon}(\xi) \sim \sum_{\ell=0}^{\infty} c^{\pm,j}_{\ell} \xi^{k_j-\ell} \ \mbox{ as } \ \xi\to \pm\infty \]
	we have:
	\[ \frac{1}{2\pi}\int_0^{2\pi}e^{-ims}\Upsilon(\theta)ds \sim \sum_{\ell=0}^{\infty} \sum_{j=1}^q c^{+,j}_{\ell} \omega_j^{-m} m^{k_j- \ell} \ \mbox{ as } \ m\to\infty. \]
\end{lem}
\begin{proof}
	Choose our cutoffs $\rho_j$ to be non-negative with disjoint supports and such that there exists $\eta\in C^{\infty}_c(\bR)$ with $0\leq \eta\leq 1$ such that
	\[ \rho_1+\cdots + \rho_q + \eta \equiv 1 \ \mbox{ on } \ [0,2\pi] \ \mbox{ and } \ \singsupp(\Upsilon)\cap \supp(\eta) = \emptyset. \]
	Then, taking Fourier transforms we have
	\[ \frac{1}{2\pi}\int_0^{2\pi}e^{-ims}\Upsilon(s)ds = \frac{1}{2\pi} \sum_{j=1}^q \int_{0}^{2\pi} e^{-ims}\rho_j(s) \Upsilon(s) ds + \frac{1}{2\pi } \int_{0}^{2\pi} e^{-ims}\eta(s)\Upsilon(s)ds. \]
	Since $\eta\Upsilon \in C^{\infty}_c(\bR)$ we have that the last term is going to $0$ rapidly as $m\to\infty$. For the remaining terms we have
	\[ \frac{1}{2\pi}\int_0^{2\pi} e^{-ims}\rho_j(s)\Upsilon(s)ds = \hat{\rho_j\Upsilon}(m) \sim e^{-is_jm}\sum_{\ell=0}^{\infty} c^{+,j}_{\ell} m^{k_j-\ell} \ \mbox{ as } \ m\to\infty.  \]
	Summing these asymptotics together then yields our desired result.
\end{proof}

So we see that in order to obtain the leading order asymptotics of $\mu(E,m,\varphi)$ as $m\to\infty$ it suffices to demonstrate that the singularities of $\Upsilon(\varphi)$ are classical and to compute both its order as an FIO, and the leading terms $c^{+,j}_{0}$ in the asymptotic expansions of its Fourier transform. Let's now check how to obtain $c^{+,j}_0$ from the principal symbol.

\begin{lem}
	Let $s_0$ be a classical singularity of degree $k$ of $\Upsilon$, let $\rho\in C^{\infty}_c(\bR)$ be a cutoff as in the previous lemma and let $a(s,\xi)$ be any principal symbol for $\rho\Upsilon$. i.e.
	\[ a(s,\xi)-e^{is_0\xi}\hat{\rho\Upsilon}(\xi) \in S^{k-1}(\bR_s\times \bR_{\xi}). \]
	Then
	\[ c^{\pm}_0 = \lim_{\xi\to\pm\infty} a(s,\xi)\xi^{-k}. \]
\end{lem}
\begin{proof}
	Indeed, if $a(s,\xi)$ is any principal symbol for $\rho\Upsilon$ then, by definition
	\[ |a(s,\xi)-e^{is_0\xi}\hat{\rho\Upsilon}(\xi)| \lesssim (1+|\xi|)^{k-1} \]
	and so dividing by $\xi^k$ and taking limits yields our desired result.
\end{proof}

So, our goal has now been reduced to writing $\Upsilon(\varphi)$ as a composition of well-understood FIOs and computing the order and principal symbol of the composition in terms of its constituents. Let's begin by introducing the relevant operators from \cite{guillemin1990reduction} and \cite{strohmaier2018gutzwiller}.

\begin{defn}
	Let $E_{adv}$ and $E_{ret}$ respectively denote the \textbf{advanced} and \textbf{retarded} fundamental solutions for $\Box_{\omega}$. Explicitly, for $f\in C^{\infty}_c(P)$, $u:=E_{adv}f$ is the unique solution to $\Box_{\omega}u=f$ whose support is contained in the forward causal set of $\supp(f)$. i.e. $u$ solves $\Box_{\omega}u=f$ with vanishing Cauchy data in the past before $\supp(f)$. Similarly, $E_{ret}f$ is the unique solution to the Cauchy problem $\Box_{\omega}u=f$ with vanishing Cauchy data in the Causal future of $f$ (in the future after $\supp(f)$).
\end{defn}

\begin{lem}{\cite{strohmaier2018gutzwiller}} \\
	let $f\in C^{\infty}(P_0)$ and let $|dV_{P_0}|$ denote the measure on $P$ given by integration over $P_0$ with respect to the induced volume measure on $P_0$ from the metric. Write $E:= E_{adv}-E_{ret}$. Then
	\[ u:= E(f_1|dV_{P_0}| +\partial_{\hat{n}}(f_2|dV_{P_0}|)) \]
	(where $\partial_{\hat{n}}(f_2|dV_{P_0}|)$ is a distributional derivative of a measure) is the unique solution to the Cauchy problem $\Delta_{\omega}u=0$ with Cauchy data
	\[ \begin{cases}
		u(x,0) &= f_1(x) \ \mbox{ on } P_0, \\
		(\cL_{\hat{n}}u)(x,0) &=f_2(x) \ \mbox{ on } P_0.
	\end{cases} \]
	Furthermore $E\in I^{-3/2}(P\times P;C_1')$ with the canonical relation $C_1$ given by
	\[ C_1 = \{(\zeta_1; \ \zeta_2) \in T^*_0P\times T^*_0P \ : \ \exists s\in \bR \mbox{ such that } \zeta_2 = G_{-s}(\zeta_1)\}. \]
	Parametrizing the left copy of $T^*_0P$ in $C_1$ by $T^*_0P|_{P_0}\times \bR \cong \cN\times \bR_{s'}$ via the geodesic flow and then the $\zeta_2=G_{-s}(\zeta_1)$ by the parameter $s$, the principal symbol of $E$ is given by the half-density
	\[ |d_{C_1}|^{1/2}:=-\frac{1}{2}|\Omega_{\cN}|^{1/2}\otimes |ds'|^{1/2}\otimes |ds|^{1/2} \]
	where $\Omega_{\cN}$ is the Liouville volume form on $\cN$ induced by the symplectic form.
\end{lem}

Before we get to composing FIO's, let's recall how this works \cite{hormander2009analysis}. Suppose we had smooth manifolds $X,Y,Z$ of respective dimensions $n_X,n_Y,n_Z$ respectively and $C_1\subseteq (T^*Z\setminus 0)\times (T^*Y\setminus 0)$, $C_2 \subseteq (T^*Y\setminus 0)\times (T^*X\setminus 0)\setminus 0$ canonical relations. We write $C_j'$ for the result of multiplying the left fiber variables by $-1$ so that the result is a Lagrangian submanifold. Given
\[ A_1 \in I^{d_1}(Z\times Y;C_1') \ \mbox{ and } \ A_2 \in I^{d_2}(Y\times X;C_2') \]
we interpret $A_1$ and $A_2$ as operators
\[ A_1:C^{\infty}_c(Y)\to \cD'(Z) \ \mbox{ and } \ A_2:C^{\infty}_c(X) \to \cD'(Y). \]
One can then often form the composition
\[ A_1\circ A_2 \in I^{d_1+d_2+\frac{e}{2}}(Z\times X;(C_1\circ C_2)') \]
where $e$ and $(C_1\circ C_2)'$ are defined as follows. Since $C_1$ and $C_2$ are Lagrangian they have dimensions:
\[ \dim(C_1) = n_X+n_Y \ \mbox{ and } \ \dim(C_2) = n_Y+n_Z. \]
The product $C_1\times C_2$ lives in $T^*Z \times (T^*Y)^{\times 2}\times T^*X$ and has dimension
\[ \dim(C_1\times C_2) = n_X +2n_Y + n_Z. \]
Meanwhile we also have a diagonal submanifold
\[ D:=(T^*Z\setminus 0)\times \diag(T^*Y\setminus 0)\times (T^*X\setminus 0) \]
of dimension $\dim(D) = 2n_Z + 2n_Y + 2n_X$. Since the total space $T^*Z\times (T^*Y)^{\times 2}\times T^*X$ has dimension $2n_Z +4n_Y +2n_X$ it follows that if $D$ and $C_1\times C_2$ intersected transversely then the intersection would have dimension
\[ \dim(D\cap (C_1\times C_2)) = n_Z+n_X \]
and if $\pi_X,\pi_Z$ are respectively the projection maps from $T^*Z\times (T^*Y)^{\times 2}\times T^*X$ to $T^*X$ and $T^*Z$ then the restriction
\[ \pi_Z\times \pi_X|_{D\cap (C_1\times C_2)}: D\cap (C_1\times C_2) \to C_1\circ C_2:= (\pi_Z\times \pi_X)(D\cap (C_1\times C_2)) \]
is a local diffeomorphism and $C_1\circ C_2$ is a Lagrangian submanifold of $(T^*Z\setminus 0)\times (T^*X\setminus 0)$. In this case, as long as everything is properly supported, we can take $e=0$ and we have
\[ A_1\circ A_2 \in I^{d_1+d_2}(Z\times X;(C_1\circ C_2)'). \]
We call this a \textbf{transverse composition} of FIO's. Furthermore, in this case if $a_1,a_2$ are the principal symbols of $A_1,A_2$ then:
\[ \mbox{the principal symbol of } A_1\circ A_2 \mbox{ is given by the restriction of } a_1\times a_2 \mbox{ to } (C_1\circ C_2)'. \]
However, one can still form the composition $A_1\circ A_2$ if the intersection of $D$ and $C_1\times C_2$ in $T^*Z\times (T^*Y)^{\times 2}\times T^*X$ is merely \textbf{clean}. In this case the intersection is still a smooth manifold, its tangent spaces are given by the intersections of the tangent spaces of $D$ and $C_1\times C_2$, $C_1\circ C_2$ is still defined in the same way, but now the projection map
\[ \pi_Z\times \pi_X|_{D\cap (C_1\times C_2)}:D\cap (C_1\times C_2) \to C_1\circ C_2 \]
is merely required to be a submersion. Since we're assuming everything is properly supported it follows that the fibers are compact manifolds. We define:
\[ e:= \mbox{ the dimension of the fibers of } \pi_Z\times \pi_X|_{D\cap (C_1\times C_2)}. \]
This is called the \textbf{excess}. Then from Proposition 25.1.5' in \cite{hormander2009analysis} we have
\[ A_1\circ A_2\in I^{d_1+d_2+\frac{e}{2}}(Z\times X;(C_1\circ C_2)') \]
where if $a_1,a_2$ are the principal symbols of $A_1,A_2$ respectively then the principal symbol of $A_1\circ A_2$ at a point $z\in (C_1\circ C_2)'$ is given by
\[ \int_{F_z}a_1\times a_2 \ \mbox{ where } \ F_z \ \mbox{ is the fiber over } z. \]
We will call this a \textbf{clean composition} of FIOs. It should be noted that the above results can be occasionally tweaked to apply when some of the hypotheses (such as $C_j$ being a canonical relation) aren't exactly satisfied as long as one is careful to ensure that the wavefront sets line up correctly in order for the desired products to be defined. \\

Finally we should say that in our below computations we omit the Maslov index factors until the very end. \\

We are now ready to apply the above FIO calculus in order to better understand our generating function $\Upsilon(\varphi)$. Let's recall our notation from earlier:
\begin{itemize}
	\item $d$ is the dimension of $G$,
	\item $n+1$ is the dimension of $M$ with $n$ the dimension of $\Sigma_0$,
	\item $n+1+d$ is the dimension of $P$,
	\item $T^*_0P$ is a cone subbundle of $T^*P\setminus 0$ and has dimension $2(n+1+d)-1$.
	\item The restriction $T^*_0P|_{P_0}$ is symplectomorphic to $T^*P_0\setminus 0$ (but not in a $\bR_{>0}$-equivariant way) and both have dimension $2(n+d)$.
	\item $\dim\cO=:2\ell$ so $\cN_{\cO}$ has dimension $2(n+\ell)$.
\end{itemize}
The below result is also from \cite{strohmaier2018gutzwiller} and again we state it for the reader's convenience.

\begin{lem}{\cite{strohmaier2018gutzwiller}} \\
	Let $E_t(x,y):= e^{-it(D_Z)_x}E(x,y)$. Then
	\[ E_t(x,y) \in I^{-7/4}(P\times P\times \bR;C_2') \]
	where $C_2$ is the canonical relation:
	\begin{align*} C_2 &:= \{ (\zeta_1; \ \zeta_2; \ t,\tau) \in (T^*_0P)^{\times 2} \times (T^*\bR\setminus 0) \\ & \ \ \ \ \ \ \ \ \ \ \ \ \ \ \  : \ \tau+\langle Z^{\omega},\zeta_1\rangle =0, \ \exists s \mbox{ such that } \zeta_2 = (G_{-s}\circ \Phi_t^Z)(\zeta_1) \}. \end{align*}
	Parametrizing $C_2\cong C\times \bR_t$ as the flowout of $C$ under the $Z$-flow the principal symbol of $E_t(x,y)$ is given by:
	\[ \mp\frac{i}{2}(2\pi)^{3/4}|d_{C_1}|^{1/2}\otimes |dt|^{1/2} \ \mbox{ on } \ C_{\pm} \]
	where $C_{\pm}$ is the subset of $C_1$ where both covectors are in $T^*_{\pm}P$.
\end{lem}

In the next lemma we begin combining results from \cite{strohmaier2018gutzwiller} and \cite{guillemin1990reduction}.

\begin{lem}{\label{lemGamma}}
	The right $G$-action gives us an action map $C^{\infty}(P)\to C^{\infty}(P\times G)$ which is an FIO
	\[ F\in I^{-d/4}(P\times P\times G;\Gamma_0') \]
	with $\Gamma_0'$ the moment Lagrangian, whose canonical relation is:
	\[ \Gamma_0 := \{ (\zeta; \  \zeta\cdot g; \ g,\eta) \in (T^*P\setminus 0)^{\times 2} \times (T^*G\setminus 0) \ : \ \mu(\zeta)=\eta \}. \]
	The composition $E_t\circ F$, denoted by $E_t(x,yg)$, arises from a transverse intersection of canonical relations and is therefore an FIO:
	\[ E_t(x,yg) \in I^{-(d+7)/4}(P\times P\times G\times \bR;\Gamma') \]
	with canonical relation
	\begin{align*}
		\Gamma &:= \{ (\zeta_1; \ \zeta_2; \ g,\eta; \ t,\tau) \in (T^*_0P)^{\times 2} \times (T^*G\setminus 0) \times (T^*\bR\setminus 0) \\ & \ \ \ \ \ \ \ \ \ \ \ \ \ \ \ : \tau+\langle Z^{\omega},\zeta_1\rangle =0, \ \mu(\zeta_2g^{-1})=\eta, \ \exists s \mbox{ such that } \ \zeta_2 = (G_{-s}\circ \Phi_t^Z)(\zeta_1)g \}
	\end{align*}
	Parametrizing $\Gamma$ by $C_2\times G\times \bR \cong C_1\times \bR_{t_1}\times G\times \bR_{t_2}$ the principal symbol of $E_t(x,yg)$ is given by:
	\[ \mp\frac{i}{2}(2\pi)^{(d+3)/4} |d_{C_1}|^{1/2}\otimes |dt_1|^{1/2}\otimes |dg|^{1/2}\otimes |dt_2|^{1/2} \ \mbox{ on } \ C_{\pm} \]
	where $|dg|$ is the volume measure on $G$ induced by our $\Ad$-invariant inner product on the tangent spaces.
\end{lem}
\begin{proof}
	The expression for the moment Lagrangian and the fact that $\Gamma_0 \in I^{-d/4}(P\times P\times G;\Gamma_0')$ is proven in \cite{guillemin1990reduction}. By construction we have
	\[ \Gamma = C_2\circ \Gamma_0 \]
	and the composition is clean so the orders of the FIOs simply add up.
\end{proof}

In \cite{strohmaier2018gutzwiller}, the distributional trace of $e^{-itD_Z}$ was expressed in terms of $E_t$ and so $E_t(x,yg)$ will play a similar role for our equivariant trace.

\begin{lem}
	Write $\hat{n}_x,\hat{n}_y$ for the Lie derivatives along the unit normal $\hat{n}$ in the variables $x$ and $y$ respectively. Then
	\[ \cF:= \hat{n}_xE_t(x,yg) - \hat{n}_y E_t(x,yg) \in I^{-(d+3)/4}(P\times P\times G\times \bR;\Gamma) \]
	with $\Gamma$ given in \ref{lemGamma}. Under the same parametrization of $\Gamma$ as in \ref{lemGamma}, the principal symbol of $\cF$ is given by
	\[ \pm\frac{1}{2}(2\pi)^{(d+3)/4}\langle \hat{n},-\rangle |d_{C_1}|^{1/2}\otimes |dt_1|^{1/2}\otimes |dg|^{1/2}\otimes |dt_2|^{1/2} \ \mbox{ on } C_{\pm} \]
	where $\langle\hat{n},-\rangle$ is the function on $\Gamma$ given by pairing the first cotangent vector with $\hat{n}$.
\end{lem}
\begin{proof}
	Differentiation is a differential operator, hence pseudodifferential operator, and so its Lagrangian is just the diagonal. Therefore differentiating an FIO does not affect the Lagrangian and merely increases the order by 1.
\end{proof}

For the next couple of lemmas we hold off on computing the principal symbols since it will be easier to directly compute the principal symbol of the wave trace after all of these compositions.

\begin{lem}
	Let $\diag:P\to P\times P$ denote the diagonal map so that pulling back along $\diag$ is an FIO
	\[ \diag^* \in I^{(n+1+d)/4}(P\times P\times P;C_3') \]
	with Lagrangian $C_3'$ where
	\[ C_3 = \{ (p,\zeta_2-\zeta_1; \ p,\zeta_1; \ p,\zeta_2) \in (T^*P\setminus 0)^{\times 3} \}. \]
	Then the composition $\diag^*\cF$ arises from a transverse intersection and is therefore an FIO
	\[ \diag^*\cF \in I^{(n-2)/4}(P\times G\times \bR;\Gamma_1) \]
	where
	\begin{align*} 
		\Gamma_1 &:= \{ (\zeta_2-\zeta_1; \ g,\eta; \ t,\tau) \in T^*P\times (T^*G\setminus 0)\times (T^*\bR\setminus 0) \\ & \ \ \ \ \ \ \ \ \ \ \ \ \ \ \ : \ \tau+\langle Z^{\omega},\zeta_1\rangle =0, \ \mu(\zeta_2g^{-1})=\eta, \ \exists s \mbox{ such that } \zeta_2 = (G_{-s}\circ \Phi^Z_t)(\zeta_1)g \\ & \ \ \ \ \ \ \ \ \ \ \ \ \ \ \ \ \mbox{ and } \zeta_1,\zeta_2\in T^*_0P \mbox{ live over the same point in } P \}
	\end{align*}
\end{lem}
\begin{proof}
	The expression for $\Gamma_1$ above is precisely the definition of $C_3\circ \Gamma$ so let's check that this is indeed a transverse composition. Notice that in the definition of $\Gamma_1$, one $\zeta_1$ and $t$ are chosen, $s$ and $g$ are uniquely determined by the requirement that $(G_{-s}\circ \Phi^Z_t)(\zeta_1)g$ must live over the same point in $P$ as $\zeta_1$. Furthermore, the constraint that there must exist $s,g$ so that $(G_{-s}\circ \Phi^Z_t)(\zeta_1)g$ lives over the same point in $P$ adds $n$ independent constraints on $\zeta_1$ since they must also live over the same point in $\Sigma_0$. This is unless $t=0$. \\
	
	So, given $0\neq t,\zeta_1$ satisfying our $n$ independent constraints: $s,g$ and therefore $\zeta_2$ and $\eta$ are completely determined. $\tau$ is directly determined by $\zeta_1$. Hence we see that there are exactly
	\[ \dim(T^*_0P)+1-n = 2(n+1+d)-1+1-n = (n+1+d)+d+1 \]
	independent directions in both the composition $C_3\circ \Gamma$ and in the fiber over the point corresponding to $0\neq t,\zeta_1$. \\
	
	In the case $t=0$ we necessarily have $s=0$ and $g=1\in G$, however the $t=0$ local is a proper submanifold of $\Gamma_1$ and the tangent space to $\Gamma_1$ at $t=0$ has $d+1$ tangent directions arising from how $s,g$ vary as we move off the $t=0$ local. Within the $t=0$ local we then have $\zeta_2=\zeta_1$ and $\tau,\eta$ are determined by $\zeta_1=\zeta_2$. While, in this case, we do have $\dim(T^*_0P)=2(n+1+d)-1$ choices for $\zeta_1$, it is the quantity $\zeta_2-\zeta_1$ that appears in $\Gamma_1$ and so the fiber coordinates of the first component of $\Gamma_1$ always vanish. Thus in both $\Gamma_1$ and in the fiber over the point corresponding to $(0,\zeta_1)$ we have
	\[ \dim(P)+d+1 = (n+1+d)+d+1 \]
	tangent directions. Hence indeed we have a transverse intersection and the order of $\diag^*\cF$ is the sum of the orders of $\diag^*$ and $\cF$.
\end{proof}

\begin{lem}
	Let $\iota:P_0\hookrightarrow P$ denote the inclusion. Pulling back along $\iota$ is an FIO
	\[ \iota^* \in I^{1/4}(P_0\times P;C_4') \]
	with Lagrangian $C_4'$ defined by
	\[ C_4 = \{(x,\zeta_1; \ x,\zeta_2) \in T^*P_0\times T^*P|_{P_0} \ : \ \zeta_2|_{TP_0}=\zeta_1\}. \]
	As in \cite{strohmaier2018gutzwiller}, the canonical relation of $\diag^*\cF$ is disjoint from the conormal bundle $N^*P_0$ and $C_4\circ \Gamma_1$ arises from a tranverse intersection so the composition $\iota^*\diag^*\cF$ can be formed as if it were a transverse composition of FIOs and
	\[ \iota^*\diag^* \cF \in I^{(n-1)/4}(P_0\times G\times \bR;\Gamma_2) \]
	where
	\begin{align*} \Gamma_2 &:= \{ ( (\zeta_2-\zeta_1)|_{TP_0}; \ g,\eta; \ t,\tau) \in T^*P_0\times (T^*G\setminus 0)\times (T^* \bR\setminus 0) \\ & \ \ \ \ \ \ \ \ \ \ \ \ \ \ \ : \zeta_1,\zeta_2 \in T^*_0P|_{P_0} \mbox{ lie over the same point in } P_0, \ \tau+\langle Z^{\omega},\zeta_1\rangle =0, \\
	& \ \ \ \ \ \ \ \ \ \ \ \ \ \ \ \ \mu(\zeta_2g^{-1})=\eta, \ \exists s \mbox{ such that } \zeta_2=(G_{-s}\circ \Phi_t^Z)(\zeta_1)g \}.
	\end{align*}
\end{lem}
\begin{proof}
	The proof that the composition $i^*\diag^*\cF$ can be formed is exactly the same as in Lemma 8.3 and the discussion preceding it in \cite{strohmaier2018gutzwiller}, and our $\Gamma_2$ is precisely defined to be $C_4\circ \Gamma_1$. Transversality again follows from noticing that the intersection is clean and then dimension-counting, however we should remark that in order to get exactly $n+2d+1$ degrees of freedom one uses the fact that the restriction of covectors in $T^*_0P|_{P_0}$ to $TP_0$ yields the isomorphism $T^*_0P|_{P_0}\cong T^*P\setminus 0$ and hence the only way the fiber variable of the first component is zero is if $\zeta_1=\zeta_2$.
\end{proof}

So, we've arrived at the following object:
\[ \iota^*\diag^*\cF = \left( \hat{n}_xE_t(x,yg) - \hat{n}_y E_t(x,yg)\right)|_{x=y\in P_0} \in I^{(n-1)/4}(P_0\times G\times \bR;\Gamma_2). \]
The importance of this object arises from the following slight generalization of Theorem 4.1 from \cite{strohmaier2018gutzwiller}.

\begin{prop}
	Let $\Pi_*:C^{\infty}(P_0\times G\times \bR)\to C^{\infty}(G\times \bR)$ be the operator given by integration over $P_0$. Then since the base has dimension $d+1$ and the fibers have dimension $n+d$ we have
	\[ \Pi_* \in I^{\frac{d+1}{2}-\frac{n+d}{4}}(G\times \bR\times P_0\times G\times \bR; C_5') \]
	where
	\begin{align*}
		C_5 &=\{ (g,\eta; \ t,\tau; \ x,0; \ g,\eta; \ t,\tau) \\ & \ \ \ \ \ \ \ \ \ \ \in (T^*G\setminus 0)\times (T^*\bR\setminus 0)\times  T^*P_0 \times (T^*G\setminus 0)\times (T^*\bR\setminus 0) \ : \ x\in P_0 \}.
	\end{align*}
	Furthermore, since $m_0=1$ and $Q_{\omega}$ is positive definite on $\cH=\bigoplus_{m\geq 1}^{L^2}\cH_m$, if we set
	\[ \cV:= \cH^{\perp Q_{\omega}} \ \mbox{ then } \ \ker\Box_{\omega} = \cV\oplus \cH \]
	then we have
	\begin{align}
	\cK(g,t):=\Pi_*\iota^*\diag^*\cF &= \int_{P_0} \left( \hat{n}_xE_t(x,yg) - \hat{n}_yE_t(x,yg)\right)\big|_{x=y} dV_{P_0}(x) \\ \label{equivariant_trace}
	&=\Tr_{\cV}(e^{-itD_Z}\circ F) + \sum_{m=1}^{\infty} \sum_{\ell\in\bZ} \mu(m,\ell)\Tr(\kappa_m(g))e^{-it\lambda_{m,\ell}}.
	\end{align}
\end{prop}
\begin{proof}
	The basic facts concerning push-forward distributions such as $\Pi_*$ can be found in section 7.1 of \cite{strohmaier2020semi} and we omit the proofs here as they are well known. \\
	
	Let's now derive the above explicit expression \ref{equivariant_trace} for $\cK(g,t)$. Indeed, by the computation in Theorem 4.1 of \cite{strohmaier2018gutzwiller}, $\cK(g,t)$ is the equivariant trace of the operator $e^{-itD_Z}\circ F$ on $\ker\Box_{\omega}$. Recalling that $\mu(m,\ell)$ is simply the multiplicity of $\kappa_m$ in the $\lambda_{m,\ell}$-eigenspace and that $F$ acts by $\kappa_m$ on this eigenspace by definition of $\cH_m$ we obtain our above expression \ref{equivariant_trace} for $\cK(g,t)$, as desired.
\end{proof}

We will now build a distribution on $P_0\times G\times \bR\times S^1$ which we will then compose with $\iota^*\diag^*\cF$ to produce $\Upsilon(\varphi)$. A key motivating fact in the below definition is the orthogonality of the functions $g\mapsto \Tr(\kappa_m(g))$ for different $m$'s. This is a well-known fact from abstract harmonic analysis (see Section 5.3 of \cite{folland2016course}, for example), however one should take care not to confuse the two distinct notions of ``character'' of a representation.

\begin{lem}{\cite{guillemin1990reduction, guillemin1989circular}} \\
	The operator $\cL_{\cO}:C^{\infty}(G)\to \cD'(S^1)$ with
	Schwartz kernel given by the distribution
	\[ \cL(e^{i\theta},g) := \sum_{m=1}^{\infty}
	\Tr(\kappa_m(g))e^{im\theta}\]
	is in
	\[ \cL_{\cO} \in I^{(1-d)/4}(S^1\times G;\Lambda_{\cO})\]
	with Lagrangian
	\[ \Lambda_{\cO} = \{ (z,r;g,r\xi) \in (T^*S^1\setminus 0)\times (T^*G\setminus 0) \ | \ \xi\in \cO, \ g\in
	G_{\xi}, \ z=\chi_{\xi}(g)\}\]
	where $\chi_{\xi}:G_{\xi}\to U(1)$ is the character associated to $\xi\in\cO$.
\end{lem}

Our final to-do before we have, at least morally, obtained a
description of $\Upsilon(\varphi)$ as a composition of FIOs is to
localize about the ray $\lambda_{m,\ell}\sim mE$ via $\varphi$.
Towards this end, we define an operator
\[ T_{\varphi,E}:C^{\infty}_c(S^1\times \bR)\to \cD'(S^1)\]
by declaring its Schwartz kernel to be given by the oscillatory
integral
\[ T_{\varphi,E}(\theta';\theta,t) := (2\pi)^{-2}\hat{\varphi}(t)
\int_{-\infty}^{\infty} ds \ e^{is(\theta'-\theta-tE)}.\]

\begin{lem}{\cite{guillemin1990reduction}} \\
	$T_{\varphi,E}\in I^{-1/4}(S^1\times S^1\times \bR;\Lambda_E)$
	where
	\[ \Lambda_E := \{(ze^{itE},r; \ z,r; \ t,rE) \in (T^*S^1\setminus 0)\times (T^*S^1\setminus 0)\times (T^*\bR\setminus 0) \ | \ r\in \bR,
	\ z\in S^1\}.\]
\end{lem}

\begin{lem}
	We can form the composition
	\[T_{\varphi,E} \circ (\cL_{\cO}\otimes \id_{\bR}) \in
	I^{-d/4}(S^1\times G\times \bR;\Theta'_{\varphi,E})\]
	where
	\[ \Theta'_{\varphi,E} = \{( \chi_{\xi}(g)e^{itE},r ; g,r\xi ; t,Er) \in T^*S^1\times T^*G\times T^*\bR \ : \ \xi\in \cO, \ g\in G_{\xi}\} .\]
	Furthermore:
	\[ (T_{\varphi,E} \circ \cL_{\cO}\otimes \id_{\bR})\cK =
	\Upsilon(\varphi).\]
\end{lem}
\begin{proof}
	The fact that this composition $T_{\varphi,E}\circ (\cL_{\cO}\otimes\id_{\bR})$ can be formed, has the above order, and the above canonical relation $\Theta'_{\varphi,E}$ is proven in \cite{guillemin1990reduction}. So, we just need to demonstrate that we do indeed obtain $\Upsilon(\varphi)$ when applying it to $\cK$. Recalling the formula \ref{equivariant_trace} for $\cK(g,t)$ we note that by \cite{strohmaier2018gutzwiller} the trace over $\cV$ still decomposes as a sum over (possibly generalized) eigenvalues of $D_Z$ counted with multiplicity, only now not all are real and some may be zero modes. Furthermore, the $G$-dependence in the trace over $\cV$ is still in the form of the $\Tr(\kappa(g))$ for $\kappa$ the representation generated by that specific (generalized) eigenvector. Indeed, while the Hilbert space inner products from the Cauchy data isomorphism are $D_Z$-invariant they are still $G$-invariant and so $\cV$ is completely decomposable since it is a unitary $G$-representation. Since characters are orthogonal with respect to the Haar measure on $G$, we obtain:
	\[ (\cL_{m_0\cO}\otimes \id_{\bR})\cK = \sum_{m=1}^{\infty} \sum_{\ell\in \bZ} e^{-it\lambda_{m,\ell}} e^{im\theta}\]
	as a distribution on $S^1\times \bR$. Finally, applying $T_{\varphi,E}$ we immediately obtain:
	\[ (T_{\varphi,E}\circ (\cL_{m_0\cO}\id_{\bR}))\cK = \sum_{m=1}^{\infty}\sum_{\ell\in\bZ} \varphi(\lambda_{m,\ell}-mE)e^{im\theta} \]
	as desired.
\end{proof}

Our next step is to understand the composition $(T_{\varphi,E}\circ (\cL_{\cO}\otimes \id_{\bR}))\cK$ as an actual Lagrangian distribution. As it turns out, it is more clear if one first computes $(T_{\varphi,E}\circ(\cL_{\cO}\otimes \id_{\bR}))\circ \Pi_*$.

\begin{lem}
	The composition
	\[(T_{\varphi,E}\circ (\cL_{\cO}\otimes\id_{\bR}))\circ \Pi_* \in I^{\frac{1}{2}-\frac{n}{4}}(S^1\times P_0\times G\times \bR; C_6') \]
	is a transverse composition of FIOs with Lagrangian determined by
	\begin{align*}
	C_6 &= \{ (\chi_{\eta}(g)e^{itE},r; \ x,0; \ g,r\eta; \ t,rE) \in (T^*S^1\setminus 0)\times T^*P_0\times (T^*G\setminus 0)\times (T^*\bR\setminus 0) \\ & \ \ \ \ \ \ \ \ \ \ \ \  \ \  \ \ \ \ : \ \eta\in\cO, \ g\in G_{\eta} \}.
	\end{align*}
\end{lem}
\begin{proof}
	This is immediate since this composition does not affect the $T^*P_0$-variables and in the $(T^*G\setminus 0)\times (T^*\bR\setminus 0)$-variables the Lagrangian for $\Pi_*$ is just the diagonal.
\end{proof}

\begin{thm}
	The clean intersection hypothesis implies that the composition of $T_{\varphi,E}\circ (\cL_{\cO}\otimes \id_{\bR})\circ \Pi_*$ and $\iota^*\diag^*\cF$ is a clean composition of FIOs with excess
	\[ e = 2(n+\ell)-2 \ \mbox{ and therefore order } \ \left(\frac{1}{2}-\frac{n}{4}\right) + \frac{n-1}{4} + \frac{e}{2} = n+\ell -1+\frac{1}{4} \]
	Thus
	\[ \Upsilon(\varphi) \in I^{n+\ell-1+\frac{1}{4}}(S^1;\fC_{E}') \]
	where
	\begin{align*}
		\fC_{E}'&= \{ (\chi_{\eta}(g)e^{itE},r)\in T^*S^1\setminus 0 \ : \ \eta\in \cO, \ g\in G_{\eta}, \ \exists \zeta\in T^*_0P|_{P_0} \mbox{ such that} \\
		& \ \ \ \ \ \ \ \ \ \ \ \ \ \ \ \exists s \mbox{ with } \zeta = (G_{-s}\circ\Phi^Z_t)(\zeta)g, \ \langle Z^{\omega},\zeta\rangle = -rE, \ \mu(\zeta g^{-1})=r\eta \}
	\end{align*}
	\end{thm}
\begin{proof}
	The main goal here is to compute the fiber over a point $(\omega,r)\in TODO$. By homogeneity we can assume $r=1$ and so the fiber is given by:
	\begin{align*}
	F_{(\omega,1)}:=\{ (x,0; \ & g,\eta; \ t,E) \in T^*P_0\times (T^*G\setminus 0)\times (T^*\bR\setminus 0) \ : \ x\in P_0, \ \exists \zeta \in (T^*_0P)_x \\ & \mbox{ such that} \exists s \mbox{ with } \zeta = (G_{-s}\circ \Phi_t^Z)(\zeta)g, \ \langle Z^{\omega},\zeta\rangle =-E, \ \mu(\zeta g^{-1})=\eta, \\ &  \chi_{\eta}(g)e^{itE}=\omega,  \mbox{ and } \eta\in\cO, \ g\in G_{\eta} \}
	\end{align*}
	Since we chose $E>0$ our constraint $\langle Z^{\omega},\zeta\rangle = -E$ implies $\zeta \in T^*_+P|_{P_0}\subseteq T^*_0P|_{P_0}$ and therefore $\zeta$ corresponds to a unique null geodesic $\gamma \in \cN$ with $\gamma(0)=x$, $H_Z(\gamma)=E$, $\mu(\gamma g^{-1})=\eta$ and $\gamma = \Phi_t^Z(\gamma)g$. Therefore $(\gamma g^{-1},\eta)\in \mu_{\cO}^{-1}(0)$ and the image of this in the quotient is a periodic orbit in $\cN_{\cO}$ with period $t$ and energy $\tilde{H}_Z =E$. \\
	
	Now, let's write $\pi:\bR\times \mu_{\cO}^{-1}(0)\to \bR\times \cN_{\cO}$ for the projection map and recall that $\fY_E\subseteq \bR\times \cN_{\cO}$ is the set of periodic orbits for the reduced flow together with their periods. If we denote
	\[ \fX := \{ (t,\gamma,\eta,g) \ : \ (t,\gamma,\eta)\in \pi^{-1}(\fY_E) \ \mbox{ and } \ g\in G_{\eta} \} \]
	then $\dim\fX = \dim \pi^{-1}(\fY_E) + \dim G - \dim \cO = 2d+2n -1$ since $\fY_E$ has dimension $2(n+\ell)-1$ where $2\ell = \dim\cO$. Note: the clean intersection hypothesis implies that $\fY_E$ is a disjoint union of smooth manifolds with the clopen subset $\{0\}\times \tilde{H}_Z^{-1}(E)$ having dimension $=2(n+\ell)-1$ and the other components having dimension at most $2(n+\ell)-1$. Furthermore the map
	\begin{align*}
		\fX &\to F_{(\omega,1)} \\
		(t,\gamma,\eta,g) &\mapsto (\gamma(0)g, g,\eta,t)
	\end{align*}
	is a submersion. This, together with the fact that the holonomy map $\Hol:\fY_E\to \Un(1)$ is locally constant, implies that we have a clean composition of FIOs. Since the only part of the derivative $\gamma'(0)$ of $\gamma$ captured in the image of our submersion $\fX\to F_{(\omega,1)}$ is $\eta = \mu(\gamma)$ it follows that the kernel of the above submersion at each point contains a $2d-2\ell$-dimension subspace of tangent vectors orthogonal to the tangent space $T_{\eta}\cO$. The only other degeneracy comes the 1-dimensional space of vectors tangent to the curve $\gamma$ itself and so we arrive at:
	\[ \dim F_{(\omega,1)} = 2(n+d)-1 - 2(d-\ell)-1 = 2(n+\ell)-2 \]
	as desired.
\end{proof}

All that remains now is the calculation of the principal symbol of $\Upsilon(\varphi)$. It's worth noticing, however, that from the expression for $T_{\varphi,E}$ we see that the actually wave front set $\WF'(\Upsilon(\varphi))$ will often be a proper subset of $\fC_E'$ depending on $\supp\hat{\varphi}$. This is due to the varying dimensions of the components of $\fY_E$ and the support of the principal symbol of $\Upsilon(\varphi)$ being constrained by $\supp\hat{\varphi}$. We compute this principal symbol now.

\begin{thm}
	We have
	\begin{align*}
		\WF'(\Upsilon(\varphi)) &\subseteq \{ (\omega,r)\in S^1 \times \bR_{>0} \ : \ \exists (T,\gamma) \in \fY_E \\
		& \ \ \ \ \ \  \ \ \ \  \ \ \ \ \ \mbox{with } T\in \supp\hat{\varphi} \mbox{ such that } \Hol_{\cO}([0,T]\ni t\mapsto \tilde{\Phi}_t^Z(\gamma)) = \omega \}
	\end{align*}
	and, under the clean intersection hypothesis, the singularities of $\Upsilon(\varphi)$ are all classical. Assuming $\hat{\varphi}(0)\neq 0$ the principal symbol at $(\omega,r)\in \WF'(\Upsilon(\varphi))$ is given by:
	\[ C_{n,d}\omega^m\hat{\varphi}(0) \Vol\left( \tilde{H}_Z^{-1}(E)\right) |r|^{n+\ell-1}|d\omega\wedge dr|^{1/2}. \]
	Here we omit the Maslov factor since in this case it can be invariantly taken to be constant on $\WF'(\Upsilon(\varphi))$.
\end{thm}
\begin{proof}
	The result concerning the wave front set will follow immediately from the calculation of the principal symbol since the constraint $T\in\supp\hat{\varphi}$ comes from the support of the principal symbol. \\
	
	We can compute a principal symbol for $T_{\varphi,E}\circ (\cL_{\cO}\otimes \id_{\bR})$ by composing their explicitly given Schwartz kernels. The Schwartz kernel for the composition is then a distribution on $S^1\times G\times \bR$ with Schwartz kernel
	\[ (\theta',g,t)\mapsto \sum_{m=1}^{\infty} (2\pi)^{-2}\hat{\varphi}(t) e^{im(\theta'-tE)} \Tr(\kappa_m(g)). \]
	Recalling that the principal symbol of $\cF$ is given by
	\[ \pm \frac{1}{2}(2\pi)^{(d+3)/4}\langle \hat{n},-\rangle |d_{C_1}|^{1/2}\otimes |dt_1|^{1/2}\otimes |dg|^{1/2}\otimes |dt_2|^{1/2} \]
	we have that the principal symbol of $\Upsilon(\varphi)$ at a point $(\omega,1) \in \fC_E'$ is given by the integral over the fiber
	\[ F_{(\omega,1)}\cong \mbox{ quotient of holonomy } \omega \mbox{ clopen subset of } \fY_E \mbox{ by the action of the flow } \tilde{\Phi}^Z \]
	of the product of the symbol of $\cF$ and the symbol of
	\[ T_{\varphi,E}\circ (\cL_{\cO}\otimes \id_{\bR})\circ \Pi_*\circ \iota^*\circ \diag^* \]
	restricted to the fiber. The symbol over a more general point $(\omega,r)\in \fC_E'$ is then obtained by homogeneity in $r$. Since $0\in \supp\hat{\varphi}$ and principal symbols are defined modulo symbols of lower order it suffices to compute this integral over the quotient of the clopen subset
	\[ \{0\}\times \tilde{H}_Z^{-1}(E)\subseteq \fY_E. \]
	In this fibered product of symbols, the pairing of the $g$ in the symbol for $T_{\varphi,E}\circ (\cL_{\cO}\otimes \id_{\bR})$ and the $g$ in the symbol for $\cF$ amounts to replacing variables in the fibers of $T^*_0P|_{P_0}\cong \cN$ with fiber variables in $\cN_{\cO}$ (here the ``fibers'' are diffeomorphic to $\cO$). Since our fibered product is just over the quotient of $\{0\}\times \tilde{H}_Z$ by the flow, the pairing of the $t$-variables in the symbol for $T_{\varphi,E}\circ (\cL_{\cO}\otimes \it_{\bR})$ and the symbol for $\cF$ simply amounts to setting $t=0$ in both symbols and multiplying by $\hat{\varphi}(0)$. The effect of restricting to $P_0$ along the diagonal $P_0\hookrightarrow P\times P$ on the fibered product of symbols (aside from replacing the volume half-density on $\Gamma$ with the one on $\tilde{H}_Z^{-1}(E)$) is to divide by the function $\langle \hat{n},-\rangle$ and multiply by a dimensional constant, hence removing the function $\langle \hat{n},-\rangle$ from our symbol expression. Therefore, denoting by $C_{n,d}$ a dimensional constant and writing $\omega = e^{i\theta'}$, we obtain the symbol over the point $(\omega,1)$ as:
	\begin{align*}
	C_{n,d}\omega^m\int_{\tilde{H}^{-1}_Z(E)} & \hat{\varphi}(0)|d\omega\wedge dr|^{1/2} \frac{1}{|\nabla \tilde{H}_Z|^2} \nabla\tilde{H}_Z \llcorner dV_{\cN_{\cO}} \\ &= C_{n,d}\omega^m \hat{\varphi}(0)\Vol\left( \tilde{H}^{-1}_Z(E)\right)|d\omega\wedge dr|^{1/2}.
	\end{align*}
	We can now recover the principal symbol over $(\omega,r)$ by scaling. Since the fibers are diffeomorphic to $\tilde{H}_Z^{-1}(E)/\bR$ where the $\bR$-action is by the Hamiltonian flow of $\tilde{H}_Z$ they have dimension $2(n+\ell)-2$ and so the principal symbol over $(\omega,r)$ is given by:
	\[ C_{n,d}\omega^m \hat{\varphi}(0)\Vol\left( \tilde{H}_Z^{-1}(E)\right) |r|^{(n+\ell)-1}|d\omega\wedge dr|^{1/2}\]
	where we note that $n+\ell -1$ is half the dimension of our fiber.
\end{proof}

\begin{thm}
	Under the assumptions of \ref{isolated_periods} where the time $\neq 0$ part of the set $\fY_E$ consists of finitely many isolated periodic orbits $(T_1,\gamma_1),....,(T_q,\gamma_q)$, and assuming $0\notin \supp\hat{\varphi}$ we actually have $\Upsilon(\varphi) \in I^{1/4}$ with principal symbol at each $(\Hol_{\cO}(T_j,\gamma_j),r)$, $j=1,...,q$, given by:
	\[ C_{n,d} \Hol_{\cO}(T_j,\gamma_j)^m \frac{T_j^{\#}}{2\pi} \hat{\varphi}(T_j) |\det(I-P_j)|^{-1/2} e^{i\pi \fm_j/4}|d\omega\wedge dr|^{1/2}  \]
	where $T_j^{\#}$ is the primitive period of $\gamma_j$, $P_j$ is the linearized Poincar\'e first return map of $\gamma_j$ with respect to any local symplectic transversal and we have included the Maslov factor $e^{i\pi\fm_j/4}$ where $\fm_j$ is the Conley-Zehnder index of $\gamma_j$ as in \cite{strohmaier2018gutzwiller}.
\end{thm}
\begin{proof}
	The proof is exactly the same as the previous one only instead of integrating over $\{0\}\times\tilde{H}^{-1}_Z(E)$ with respect to its invariant measure we integrate over the respective periodic orbit $\gamma_j$ with respect to the density from \ref{isolated_periods}.
\end{proof}

These last three theorems respectively conclude the proofs of Theorems \ref{thm_1},\ref{thm_2} and \ref{thm_3}.


\bibliography{bibliography}
\bibliographystyle{plain}

\end{document}